%% file: flows_main.tex
\documentclass[10pt]{article}
\usepackage[bindingoffset=0.1cm, left=1.2cm, right=3cm, top=1.5cm, bottom=1.cm, 
marginparsep=0.2cm, includehead, includefoot]{geometry}
\input{preamble}
\input{notation}

	\title{Minimum-Peak-Cost Flows Over Time}
	
	\author[1,*]{Mariia Anapolska}
	\author[2]{Emma Ahrens}
	\author[1]{Christina B\"{u}sing}
	\author[1]{Felix Engelhardt}
	\author[1]{Timo Gersing}
	\author[1]{Corinna Mathwieser}
	\author[1]{Sabrina Schmitz}
	\author[1]{Sophia Wrede}
		\affil[1]{\small Combinatorial Optimization, RWTH Aachen University, 
			Germany}
		\affil[2]{Software Modeling and Verification, RWTH Aachen University, Germany}
		\affil[*]{Corresponding author: \texttt{anapolska@combi.rwth-aachen.de}}
		
	\date{}

\begin{document}
\maketitle

	\begin{abstract}
		When planning transportation whose operation requires non-consumable resources, the peak demand for allocated resources is often of higher interest than the duration of resource usage.
		For instance, it is more cost-effective to deliver parcels with a single truck over eight hours than to use two trucks for four hours, as long as the time suffices.
		To model such scenarios, we introduce the novel minimum peak cost flow over time problem, whose objective is to minimise the maximum cost at all points in time rather than minimising the integral of costs. 
		We focus on minimising peak costs of temporally repeated flows. 
		These are desirable for practical applications due to their simple structure. 
		This yields the minimum-peak-cost temporally repeated flow problem (\genProb).
		
		We show that the simple structure of temporally repeated flows comes with the drawback of arbitrarily bad approximation ratios compared to general flows over time.
		Furthermore, our complexity analysis shows the integral version of \genProb is strongly $\np$-hard, even under strong restrictions.
		On the positive side, we identify two benign special cases: unit-cost series-parallel networks and networks with time horizon at least twice as long as the longest path in the network (with respect to the transit time). 
		In both cases, we show that integral optimal flows if the desired flow value equals the maximum flow value and fractional optimal flows for arbitrary flow values can be found in polynomial time.
		For each of these cases, we provide an explicit algorithm that constructs an optimal solution.
	\end{abstract}

	\input{sections/introduction}
	\input{sections/preliminaries}
	\input{sections/reduction}
	\input{sections/unitCostIntro} 
	\input{sections/earlArrFlows}
	\input{sections/ssp-algo}
	\input{sections/generalLongT}

	\input{sections/outlook}
	\newpage
	\bibliographystyle{plainurl}
	\bibliography{references}
\end{document}

%% file: preamble.tex
\usepackage{fullpage}
\usepackage{caption}
\usepackage{subcaption}
\usepackage{amsmath,amsfonts,amssymb}
\usepackage{amsthm}
\usepackage[T1]{fontenc}
\usepackage{etoolbox}
\usepackage{textcomp}
\usepackage{optidef}
\usepackage[ruled,linesnumbered,noend,vlined]{algorithm2e}
\SetKwInput{KwData}{Input}
\SetKwInput{KwResult}{Output}
\usepackage{graphicx}
\usepackage{paralist}
\usepackage{comment}
\usepackage{bm}
\usepackage{url}
\usepackage{graphics}
\usepackage{multirow,longtable}
\usepackage{rotating}
\usepackage[dvipsnames]{xcolor}
\usepackage{booktabs}
\usepackage{pdflscape}
\usepackage{tikz}
\usepackage{float}
\usepackage[normalem]{ulem}
\usepackage{nicefrac}
\usepackage{color} 
\usepackage{pgfplots}
\usepackage{booktabs}
\usepackage{pifont}
\pgfplotsset{compat = newest}
\usepackage{xspace}
\usepackage{mathtools}
\usepackage{environ}
\usepackage{lineno} 
\usepackage{setspace} 
\usepackage{thmtools}
\usepackage{thm-restate}
\usepackage{hyperref}
\usepackage[capitalize,noabbrev]{cleveref}
\usepackage{authblk} 
\usepackage{subfiles}
\usepackage{bbm}
\usepackage{enumitem}
\usepackage[textwidth=2.5cm]{todonotes}

\usetikzlibrary{calc, decorations.pathreplacing, arrows}

\input{mythmstyles}
\input{tikz/tikz-preamble}

\renewcommand{\geq}{\geqslant}
\renewcommand{\leq}{\leqslant}

\allowdisplaybreaks
\newif\ifhideproofs
\ifhideproofs
\usepackage{environ}
\NewEnviron{hide}{}

\fi

%% file: mythmstyles.tex
\makeatletter
\def\th@plain{%
	\thm@notefont{}
	\itshape 
}
\def\th@definition{%
	\thm@notefont{\normalfont\itshape}
	\normalfont 
}
\makeatother

\newtheoremstyle{mydefinition}
{5pt}
{2ex}
{\setlength{\parskip}{0.7ex}}
{}
{\bfseries}
{.}
{\newline}
{\thmname{#1}\thmnumber{ #2}\thmnote{ \normalfont({\itshape{#3}})}}
\newtheoremstyle{mynotes}
{5pt}
{2ex}
{\setlength{\parskip}{0.7ex}}
{}
{\bfseries}
{.}
{0.5em}
{}

\newtheoremstyle{mytheorem}
{10pt}
{0pt}
{\itshape}
{}
{\bfseries}
{.}
{0.5em}
{}

\newenvironment{claimproof}[1]{\par\noindent\underline{Proof:}\space#1}{\hfill $\blacksquare$}

\theoremstyle{mytheorem}
\newtheorem{theorem}{Theorem}

\newtheorem{lemma}[theorem]{Lemma}
\newtheorem{corollary}[theorem]{Corollary}

\theoremstyle{mydefinition}

\newtheorem{definition}{Definition}

\theoremstyle{mynotes}

\newtheorem{remark}[theorem]{Remark}

\newtheorem{claim}{Claim}
\usepackage{etoolbox}
\AtEndEnvironment{proof}{\setcounter{claim}{0}}

%% file: tikz/tikz-preamble.tex
\usepackage{tikz}
\usetikzlibrary{arrows.meta, shapes, arrows, automata, fit, positioning, calc, decorations.pathmorphing, decorations.shapes, decorations.pathreplacing}
\usepackage[dvipsnames]{xcolor}

\definecolor{cost}{rgb}{0.95,0.45,0}
\definecolor{capa}{rgb}{0,0,0.7}
\definecolor{transit}{RGB}{74,127,203}

\definecolor{combi-cyan}{RGB}{0,170,170}
\definecolor{combi-darkcyan}{RGB}{0,100,100}
\definecolor{combi-orange}{RGB}{250,170,30}
\definecolor{combi-green}{RGB}{180,210,50}

\newcommand{\tauuc}[5]{
    edge[#5] node[yshift=-0.5mm, #4] {{%
        \textcolor{transit}{#1}%
        \if\relax\detokenize{#2}\relax%
            \else%
            \if\relax\detokenize{#1}\relax\else{\,/\,}\fi%
        \fi%
        \textcolor{capa}{#2}%
        \if\relax\detokenize{#3}\relax%
            \else%
            \if\relax\detokenize{#2}\relax%
                \if\relax\detokenize{#1}\relax%
                    \else{\,/\,}\fi%
                \else{\,/\,}\fi%
        \fi%
        \textcolor{cost}{#3}}}
}

%% file: notation.tex
\newcommand{\tp}{\ensuremath{\theta}\xspace} 
\newcommand{\T}{\ensuremath{T}\xspace} 
\newcommand{\auxtp}{\xi} 
\newcommand{\extp}{\hat\tp} 
\newcommand{\ua}[1][a]{u_{#1}} 
\newcommand{\ca}[1][a]{c_{#1}} 
\newcommand{\xa}[1][a]{x(#1)} 
\newcommand{\trt}[1][a]{\tau_{#1}} 
\newcommand{\trtp}[1][p]{\trt[](#1)} 
\newcommand{\val}[1][f]{\abs{\,#1}} 
\newcommand{\auxcost}{\gamma} 
\newcommand{\cmax}[1][f]{c^{\text{max}}(#1)} 
\newcommand{\suma}{\sum_{a \in A}} 
\newcommand{\auxG}{\overline{G}} 
\newcommand{\net}{\ensuremath{\mathcal{N}}\xspace} 
\newcommand{\netwDef}{(G, \ua[], \trt[], \ca[])}
\newcommand{\dem}{\ensuremath{D}\xspace}

\newcommand{\paths}{\ensuremath{\mathcal{P}}}
\newcommand{\pospaths}[1][\pflowvar]{\paths^{+}_{#1}}
\newcommand{\cycles}{\ensuremath{\mathcal{C}}}
\newcommand{\pflowvar}{y}
\newcommand{\pflow}[1][p]{\pflowvar(#1)}
\newcommand{\flowdecomp}{\pflowvar \colon\paths \cup \cycles \to  \Q_{\geq 0}}
\newcommand{\pathdecomp}{\ensuremath{\pflowvar \colon\paths \to \Q_{\geq 0}}\xspace}
\newcommand{\actpaths}[2][a]{\paths_{#1}(#2)}
\newcommand{\restr}[2][p]{#1\vert_{[#2]}}

\newcommand{\instantcosts}{cost at a time point\xspace}
\newcommand{\peakcosts}{peak cost\xspace}
\newcommand{\minpeakcosts}{minimum-peak-cost\xspace}
\newcommand{\shortdec}{\T-bounded\xspace} 
\newcommand{\TR}{temporally repeated\xspace}
\newcommand{\genProb}{\textsc{MPC-TRF}\xspace}
\newcommand{\SubSum}{\textsc{SubsetSum}\xspace}

\newcommand{\abs}[1]{\left|#1\right|}
\newcommand{\bigabs}[1]{\big|#1\big|}

\newcommand{\N}{\ensuremath{\mathbb{N}}}
\newcommand{\Z}{\ensuremath{\mathbb{Z}}}
\newcommand{\Q}{\ensuremath{\mathbb{Q}}}
\newcommand{\R}{\ensuremath{\mathbb{R}}}

\newcommand{\np}{\ensuremath{\mathcal{NP}}\xspace}

\newcommand{\fromto}[2][1]{\{#1,\ldots,#2\}}
\newcommand{\oneto}[1]{[#1]}
\newcommand{\zeroto}[1]{\fromto[0]{#1}}
\newcommand{\inst}{\ensuremath{\mathcal{I}}\xspace}
\newcommand{\iinst}{\ensuremath{\widetilde{\inst}}\xspace} 
\newcommand{\indic}[1]{\mathbbm{1}_{#1}} 
\newcommand{\flo}[1]{\ensuremath{\left\lfloor{#1}\right\rfloor}}

\renewcommand{\phi}{\varphi}

\newcommand{\True}{\textnormal{True}\xspace}
\newcommand{\False}{\textnormal{False}\xspace}



%% file: sections/introduction.tex
\section{Introduction}
Network flows are one of the fundamental models in operations research~\cite{AMO1993,FFbook,Williamson_2019}. 
In most of the models the flows are considered to be static; however, in many applications, time plays a crucial role. 
To integrate the temporal aspect, traditional (capacitated) flow networks can be extended by transit times, which describe the time that flow particles need to traverse an arc~\cite{FF1958, FFbook}.
The resulting models are called \emph{flows over time}, or \textit{dynamic flows} in some early literature.
Similarly to the traditional case, we can also extend flows over time by arccosts, leading to the \emph{min-cost flow over time} problem. 
In general, for min-cost flow over time, cost is measured as the sum over the costs at each time step~\cite{Skutella2009}.

In this work, we propose an alternative objective for min-cost flows over time, namely \emph{min-peak-cost flows over time}.
The motivation for this work originated from bed transports in a hospital, which can be modelled as flows over time.
The effective cost of bed transportation does not depend on the number of patients transported, but on the peak number of staff needed in each shift to perform bed transports. 
That means that the maximum amount of staff needed for transports simultaneously, the \emph{\peakcosts}, is the objective to be minimised. 
For this, we model the required workforce per unit of flow along an arc of a graph by arc costs. The total amount of workforce required at a certain moment during the transportation process is then described by the notion of \emph{\instantcosts}. 
A flow over time that minimises the maximum \instantcosts over the time horizon of the flow, i.e.~the \peakcosts, thus reduces the workforce needed to be reserved for a system.
A similar setting arises whenever some type of transportation, for instance, public transport, is modelled as a flow over time: the \peakcosts determine the minimum amount of resources, e.g.~busses, that need to be available to solve a given transportation problem.
This setting is especially relevant in the scope of disaster relief: naturally, a quick reaction to disasters is desirable; however, scarce resources such as staff and means of transport are usually the bottleneck for transportation planning in disaster areas.

In this contribution, we first formally introduce the problem and then derive complexity results for both the general setting and two special instance classes, as outlined in \Cref{sec:contribution}.

Solutions for instances of flow over time problems might lack helpful structures. 
However, especially if the transportation plan has to be memorised and executed by humans or primitive machines, simple and comprehensive solutions become more relevant. 
One class of flows over time with an intuitive structure and a compact description of solutions are \emph{temporally repeated flows}. 
Here, we are allowed to choose a set of paths connecting the source and sink and have to stick with this choice for the entire time horizon.
Due to their favourable structure, we focus on temporally repeated solutions in this work.

\subsection{Related work} \label{sec:related_work}
{Flows over time, or dynamic flows, were first introduced by Ford and Fulkerson~\cite{FF1958, FFbook}, who established} the maximum dynamic flow problem.  
{The computational complexity of dynamic flow problems depends on the choice of objectives and the existence of arc costs, as we see next.}

\paragraph*{Maximum and quickest flows} 
Ford and Fulkerson show in their seminal work that a flow over time of maximum value is computed in polynomial time \cite{FF1958}. 
In the quickest flow over time problem, the objective is to minimise the arrival time, i.e.~the makespan, for a given flow value; the problem is also solvable in strongly polynomial time~\cite{BurkardDK93,FleischerTardos98}. 
Well-studied extensions of this problem are 
the quickest transshipment problem~\cite{HT2000}, lexicographic flows \cite{Skutella2024,HT1987} and earliest arrival flows~\cite{Gale1959, Minieka73}.
The first two problems admit exact polynomial algorithms, the earliest arrival flow problem an FPTAS;
Skutella gives a more detailed overview~\cite{Skutella2024}.

These algorithms were originally obtained for the discrete time model introduced by Ford and Fulkerson, in which the time is measured in discrete steps of length one.
Fleischer and Tardos introduce a continuous counterpart to the time model and transfer several exact algorithms and approximation schemes to work in the continuous model as well \cite{FleischerTardos98}.

\paragraph*{Flows over time with costs} When arc costs are added to the network, already the minimum-cost maximum flow over time problem is \np-hard, as is finding a minimum-cost maximum \TR flow~\cite{KW04}.
{However}, the minimum-cost flow problem admits an FPTAS~\cite{FS2003}.
Somewhat surprisingly, flipping the objective leads to the polynomial-time solvable quickest minimum cost transshipment problem~\cite{SKUTELLA2023_quickMinTP}.

For bi-objective optimisation of cost and travelling time, Parpalea and Ciurea propose a pseudo-polynomial algorithm~\cite{ParpaleaCiurea11}.
The maximum energy-constrained flow problem, where each node has a bound on the total amount of flow passing through it, is a special case from the complexity theory point of view: not only is the integral decision problem strongly \np-complete, but the optimisation problem is also APX-hard~\cite{BodlaenderTDL08}.
However, for the general, fractional case an FTPAS exists \cite{FeketeHKK08}, and the problem can be solved in (pseudo)-polynomial time for graphs with bounded tree width \cite{BodlaenderTDL08} or uniform transit times \cite{FeketeHKK08}.  
Still, finding an exact solution is generally \np-hard, and solutions using a path representation may require an exponential number of paths~\cite{FeketeHKK08}.

\paragraph*{Temporally repeated flows} A reoccurring challenge in dynamic flow problems is  that solutions may consist of an exponential number of paths, and within each of these paths, flow may take an arbitrary number of different values. 
In this context, Ford and Fulkerson introduce \TR flows -- a special class of flows over time distinguished by a compact representation~\cite{FF1958}. 

Temporally repeated flows can be used to realise maximum and quickest flows~\cite{FF1958, BurkardDK93}. They present "structurally easier solution[s]" to the quickest transshipment problem~\cite{SchloeterSk17}. 
Fleischer and Skutella also use \TR flows to construct a $2$-approximate solution
for quickest transshipment with costs~\cite{FleischerS07-quickest}.
However, for the min-cost maximum flow problem, \TR flows are sub-optimal, and finding them is strongly \np-hard~\cite{KW04}. 

Furthermore, finding robust maximum flows for networks with uncertain transit times is \np-hard in general. An optimal robust \TR flow, in contrast, can be found in polynomial time if the time horizon is sufficiently long~\cite{GottschalkKLPSW18}.
Finally, \TR-flows are a 2-approximation for maximum flow with load-dependent transit times~\cite{KoehlerS05}.

A broader class of problems admits varying, time- or flow-dependent capacities and transit times, as well as flows with infinite time horizons. These research areas are beyond the the scope of this paper. 
For a more detailed overview, see the surveys  \cite{aronson89, POWELL1995}.

\subsection{Our contribution}\label{sec:contribution}
As mentioned in the introduction, minimising peak costs is a reoccurring theme in transportation. 
Nevertheless, to the best of our knowledge, there has been no research on this type of objective in the context of flows over time, and of temporally repeated flows in particular.
We initiate the study of this field by introducing a first formal definition of the Minimum-Peak-Cost Flow problem (MPCF) in~\Cref{sec:preliminaries}.

This work focuses on finding \TR flows with minimum \peakcosts, which we call the \genProb problem.
We show in \Cref{sec:np-hard} that the integral problem is \np-hard already for the maximum flow value on series-parallel networks with simple arc parameters: unit transit times and capacities, and costs with values either zero or one. 
This result is tight in the sense that fixing arc costs to one for all arcs leads to a polynomial-time algorithm for maximum flows on series-parallel graphs, as presented in Section~\ref{sec:unitCosts}. 
This algorithm emerges from a relation between \genProb for maximum demand and the earliest-arrival flow problem.  
In addition, we are able to show that the same strongly-polynomial-time algorithm finds optimal fractional solutions for arbitrary demand on series-parallel graphs.
Analogously, increasing the time horizon leads to a second polynomial-time solvable case discussed in Section~\ref{sec:longT}.
Here we modify the method of Ford and Fulkerson for maximum flows and adjust it to our objective of minimum peak cost.
In both mentioned special cases, we show that the integral version of the problem for general, not maximum, target flow value remains \np-hard, while fractional solutions can be constructed in polynomial time via column generation. 
\Cref{sec:conclusion} gives a summary and
an outlook on further research.

This work is an extended version of the paper appeared in the INOC2024 special issue of Networks~\cite{Anapolska25}, which is itself partially based on Section~5.2 of the master's thesis of Ahrens~\cite{AhrensMA}.

%% file: sections/preliminaries.tex
\section{Preliminaries and definitions}\label{sec:preliminaries}

In this section, we first discuss important notation and preliminaries for flows over time. 
Afterwards, we give a formal definition for \genProb as well as some of its immediate properties.

\subsection{Notation and preliminaries}
For an integer $n\in \N$, we denote by $\oneto{n}$ the set $\{1,\ldots,n\}\subseteq \N$.
Throughout this work, let $G = (V,A)$ be a digraph with node set $V$ and arc set $A \subseteq V^2$. For a node $v \in V$, we denote by $\delta^+(v)$ the set of outgoing arcs and by $\delta^-(v)$ the set of ingoing arcs of $v$. 
A (simple) \emph{path} is a sequence $p=(v_1,\ldots,v_k)$ of pairwise distinct nodes $v_1,\ldots,v_k \in V$ such that two subsequent nodes are adjacent, i.e.~$(v_i,v_{i+1}) \in A$ for $i=1,\ldots,k-1$. 
We use the notation $p\vert_{(v_i,v_j)}$ for $i<j$ to denote the sub-path $(v_i\ldots,v_j)$ of~$p$ between~$v_i$ and~$v_j$. 

We assume that every graph has a distinguished source $s\in V$ and sink~$t\in V$. 
Then, we denote the set of all $s$-$t$~paths in~$G$ with~$\mathcal{P}$ and the set of all cycles with~$\mathcal{C}$.
Moreover, each arc~$a\in A$ is equipped with a capacity~$\ua\in \N$ and a cost~$\ca \in \N_0$.
Before we continue to define networks over time and flows over time, note that we say \emph{static} flow in order to refer to a classical $(s,t)$-flow~$x$ (without a time component) in a network~$(G,u,c)$. We write $\val[x]$ to denote the value of~$x$ and we use the notation~$\flowdecomp$ ($\pathdecomp$) to describe a flow decomposition (path decomposition) of $x$.

In networks over time, we have an additional arc property $\trt \in \N$ called \emph{transit time}.
For a path $p$ in graph~$G$, we define its transit time $\trtp \in \N$, also referred to as \emph{length}, as the sum of the transit times of all arcs of the path, i.e.~$\trtp \coloneqq \sum_{a\in p} \trt$.
We call a graph~$G$ together with the three arc functions $\ua[]$, $\trt[]$ and~$\ca[]$ a \emph{network (over time)} and write $\net= \netwDef$.

There are two common time models used to define flows over time: the discrete and the continuous model. 
In the former model, a flow unit is compactly transported, i.e.~the unit departs as a whole at the one point in time and arrives as a unit too.
In the latter model, the flow is viewed as a collection of infinitesimal particles that are injected into the network at some rate and follow each its own trajectory. 
A unit of flow is the set of particles injected into the network during one unit of time.
The two models are to a great extent equivalent for combinatorial problems \cite{FleischerTardos98}.
We follow the more recent contributions and use the continuous time model in this work.

Given a network $\net= \netwDef$, we define a \emph{flow over time} as follows: 

\begin{definition}[Flow over time \cite{Skutella2009}]\label{def:fot}
	Let $\net= \netwDef$ be a network over time with distinguished terminal vertices $s,t \in V$. 
    An $(s,t)$-flow over time~$f$, from now on called \emph{flow over time} for short, in~$\net$ with time horizon~$\T \geq 0$ consists of Lebesgue-measurable functions~$f_a:\mathbb{R}_0^+ \rightarrow \mathbb{R}_0^+$ for each~$a \in A$, where $f_a(\theta)=0$ for all $\theta > \T- \trt$. 
    The function $f_a$ represents the inflow rate into the arc~$a$ at its head. 
    Furthermore, the flow rates satisfy the following constraints.
    \begin{itemize}
    \item Capacity constraint
		\[
		0 \leq f_a(\theta) \leq \ua \quad\text{ for all } a\in A,\, \theta \in [0,\T);
		\]
	\item Weak flow conservation 
  
		\[ \sum_{a\in \delta^-(v)}\int_{0}^{\tp - \trt} f_a(\auxtp) \,d\auxtp - \sum_{a\in \delta^+(v)}\int_{0}^{\tp} f_a(\auxtp) \,d\auxtp \geq 0 \quad\text{ for all } v \in V\setminus\{s,t\},\theta \in [0,T).
		\]
	\end{itemize}
	\end{definition}
The value of a flow over time is defined as follows. 
\begin{definition}[Value of a flow over time \cite{Skutella2009}]\label{def:flowVal}
		Let $\net= \netwDef$ be a network over time and let~$f=(f_a)_{a \in A}$ be a flow over time in~$\net$ with time horizon~$T \geq 0$.  The \emph{value} of~$f$ is given by the expression
		\[
		\val\coloneqq\sum_{a\in \delta^+(s)}\int_{0}^{T} f_a(\auxtp)\,d\auxtp - \sum_{a\in \delta^-(s)}\int_{0}^{T-\trt} f_a(\auxtp)\,d\auxtp. \]
\end{definition}

A flow over time in a network is called \emph{maximum} for a given time horizon $\T$ if it has the maximum value among all flows over time with time horizon $\T$. 

Temporally repeated flows are a special type of flows over time, in which a static flow is sent repeatedly along the components of its flow decomposition as long as the time horizon allows. More precisely, \emph{temporally repeated flows} are defined as follows.
	\begin{definition}[Temporally repeated flow; \cite{Skutella2009}]\label{def:trf}
		Let $x$ be a static flow and $\flowdecomp$ its flow decomposition.
		The corresponding \TR flow with time horizon \T is defined by the flow rates
		\[
		f_a(\tp) \coloneqq \sum_{p\in \actpaths{\tp}} \pflow \qquad \text{for } a\in A,\ \tp \in [0,\T),
		\]
		where 
		\[
		\actpaths{\tp} \coloneqq \{p\in \paths \mid a = (v,w)\in p \text{ and } 
		\trt[](p\vert_{s,v}) \leq \tp \text{ and }
		\trt[](p\vert_{v,t}) < \T- \tp
		\}
        \]
        is the set of paths of the decomposition that contain arc $a$ and can transport flow over $a$ at time $\tp$ without violating the time horizon.
		For $\tp \notin [0,\T)$ we set $f_a(\tp) = 0$ for all $a \in A$.
	\end{definition}
	Note that cycles and paths of length greater than \T from a flow decomposition are ignored when constructing a corresponding \TR flow.
	
   The intuition behind \TR flows is better captured by the following alternative, path-based representation. 
   \begin{remark}\label{def:chainFlow}
   A \TR flow $f$ corresponding to a path decomposition $\pathdecomp$  and for a time horizon $\T$ is a sum of \emph{chain flows}
   $f = \sum_{p\in \paths,\,\trtp \leq \T} f^{\T}_p$,
   where a chain flow~$f^\T_p$ sends the flow at rate~$\pflow$ into a path~$p$ during the time interval $\big[0, \T-\trtp\big)$.
   \end{remark}
The value of a temporally repeated flow can be expressed via the value of its underlying static flow, as the following lemma shows.
	\begin{lemma}[\cite{Skutella2009}]\label{lem:TRflowValue}
		Let $x$ be a feasible static flow in a network $\net$ with flow decomposition $\flowdecomp$ such that $\pflow = 0$ for all $p \in \paths$ with $\trtp > \T$ and for all~$p\in \mathcal{C}$. 
		Then the value of the corresponding temporally repeated flow $f$ is
		\[\val = \T \cdot\val[x] - \suma \trt \cdot \xa.\] 
		In particular, the value of the flow over time $f$ does not depend on the chosen path decomposition of the static flow $x$.
	\end{lemma}
The condition that only paths with transit time of at most $\T$ can have positive flow values is crucial for Lemma~\ref{lem:TRflowValue}, since flow on longer paths would not reach the sink but would contribute negatively to the formula~\cite{Skutella2009}.
We refer to path decompositions that respect the time restriction as \emph{\T-bounded}.
\begin{definition}[\shortdec path decomposition]
	We call a path decomposition~$\pathdecomp$ of a static flow in a flow network \emph{\shortdec} for a time horizon~$\T \in \N$ if all flow-carrying paths, i.e.~paths $p\in \paths$ with $\pflow>0$, have length at most $\T$.
\end{definition}

Ford and Fulkerson show that maximum temporally repeated flows are maximum flows, and that they can be computed in polynomial time by the following algorithm~\cite{FF1958}. 
\begin{theorem}[\cite{FF1958}]\label{thm:FoFu}
    The following algorithm computes a maximum flow over time for a network $\net=\netwDef$ and a time horizon $\T$.
    \begin{enumerate}
        \item Construct an \emph{extended network} $\net'$ from $\net$ by adding an arc $(t,s)$ with $\ua[(t,s)]=\infty$ and $\trt[(t,s)]= -\T$.
        \item Compute a minimum cost circulation in $\net'$ with respect to arc costs $\trt$;
        extract the corresponding static ($s,t$)-flow $x$ in $\net$.
        \item Compute a flow decomposition $\flowdecomp$ of $x$.
        \item Return the \TR flow induced by the decomposition~$\pflowvar$.
    \end{enumerate}
\end{theorem}
The flow decomposition attained in Step $3$ is in fact a $\T$-bounded path decomposition.
Theorem~\ref{thm:FoFu} implies that the maximum flow value is attained by \TR flows, and that maximum flows are computed in polynomial time. 

\subsection{Problem statement and properties}
We seek to find a flow over time of a required value while keeping the cost caused by the flow small for each point in time. 
The cost of a flow at a time point $\tp$ is the accumulated amount of flow present in the network at time $\tp$, weighted for each arc $a$ by its cost coefficient $\ca$. 
More precisely, the \emph{cost at a time point} is defined as follows.  
\begin{definition}[Cost at a time point]
	
	Let $\net = \netwDef$ be a network and $f$ a flow over time with time horizon~\T.
		For a time point $\tp\in [0,\T)$, the \instantcosts $\tp$ is
		\[
		c(f, \tp) \coloneqq \sum_{a\in A}\ca \cdot \left(\int_{\tp - \trt}^{\tp} f_a(\auxtp)\,d\auxtp \right).
		\] 
\end{definition}

We seek to minimise the \emph{\peakcosts} of a flow $f$, which is the maximum cost of the flow over the time horizon~$[0,\T)$, i.e.
\[
	\cmax[f] \coloneqq \max_{\tp\in[0,\T)} c(f, \tp).
\]

Given a network $\net$, a time horizon $\T$ and a demand \dem, we refer to the problem of finding a flow over time with time horizon \T of value at least \dem with minimum peak cost as \emph{Minimum-Peak-Cost Flow} (MPCF). 

We are particularly interested in \emph{\TR} flows as solutions because of their sparse structure and compact representation. 
In the remainder of this work, we consider a variant of MPCF that seeks to minimise the \peakcosts exclusively on \TR flows. 
Next, we give a precise problem definition of \genProb and three observations on the properties thereof.
\begin{definition}[\genProb]
	An instance of \emph{Minimum-Peak-Cost Temporally Repeated Flow} (\genProb) problem consists of a network $\net = \netwDef$ with a distinguished source $s$ and sink $t$, of a time horizon $\T\in \N$ and of a demand $\dem\in \N$.
	\genProb asks for a \TR flow~$f$ in $\net$ with horizon $\T$ and value $\val \geq \dem$ that minimises the \peakcosts.
\end{definition}

Recall that we can compute in polynomial time the maximum possible amount of flow $D^{\text{max}}$ that can be routed through a given network within a given time horizon.
Therefore, we assume in the following that~$\dem\leq D^{\text{max}}$, as otherwise the instance is infeasible.

\begin{remark}\label{rem:pathDecomp}
    The \peakcosts of a \TR flow depends not only on the underlying static flow, but also on the chosen path decomposition.
\end{remark}
\begin{proof}
Temporally repeated flows resulting from different path decompositions of the same static flow and with the same time horizon can have different \peakcosts{}s, as an example in Figure~\ref{fig:pathDecompDiff} demonstrates.

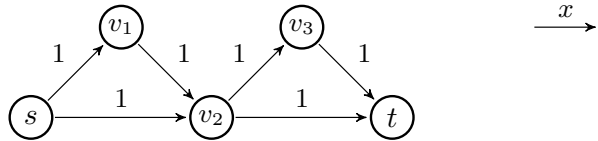
\begin{figure}[hbt]
    \centering
    \resizebox{0.5\textwidth}{!}{
        \input{tikz/path_decomp}
    }
    \caption{An instance of \genProb with unit capacities, transit times and costs.
    The displayed static flow $x$ admits two different path decompositions.
    }
    \label{fig:pathDecompDiff}
\end{figure}
Consider the network shown in the figure, with time horizon
    $T=6$ and demand $\dem = \dem^{\text{max}} = 6$.
    As we will prove in \Cref{sec:unitCosts}, the peak cost for any \TR flow in this network is attained at time $\tp = \frac{\T}{2} = 3$, and every chain flow along a path $p$ with flow rate $\pflowvar(p)$ contributes the cost $\pflowvar(p)\cdot\max\{ \trtp,\, \T-\trtp\}$ to the total \peakcosts.

Consider the first path decomposition \[
    \pflowvar \colon \paths(G)\to \R,\qquad \begin{array}{ll}
         p_1 \coloneqq (\begin{tikzpicture}[scale=0.2]
             \path[draw, ->] (0,0)--(1,1) -- (2,0) -- (3,1) -- (4,0);
         \end{tikzpicture})
             =  (s,v_1,v_2,v_3,t)&\mapsto 1,  \\
         p_2 \coloneqq 
         (\begin{tikzpicture}[scale=0.2]
             \path[draw] (0,0)--(1.75,0) node[]{};
             \path[draw, ->] (2.25,0) -- (4,0);
         \end{tikzpicture})
         = (s,v_2,t)&\mapsto 1, \\
         p &\mapsto 0 \ \text{ otherwise}.
    \end{array}
\]
The corresponding \TR flow~$f$ consists of two nontrivial chain flows with total \peakcosts
    \[
    \cmax = c(f, 3) = \pflow[p_1] \cdot (\T - \trtp[p_1]) + \pflow[p_2] \cdot \trtp[p_2] = (6-4) + 2 = 4.
    \]
For the second path decomposition 
\[
    \pflowvar' \colon \paths(G)\to \R,\qquad \begin{array}{ll}
         p_3 \coloneqq
         (\begin{tikzpicture}[scale=0.2]
             \path[draw, ->] (0,0)--(1,1) -- (2,0) -- (4,0);
         \end{tikzpicture})
         = (s,v_1,v_2,t)&\mapsto 1,  \\
         p_4 \coloneqq
        (\begin{tikzpicture}[scale=0.2]
             \path[draw, ->] (0,0)-- (2,0) --(3,1) -- (4,0);
         \end{tikzpicture})
         = (s,v_2,v_3,t)&\mapsto 1, \\
         p &\mapsto 0 \ \text{ otherwise},
    \end{array}
\]
the corresponding \TR flow $f'$ has \peakcosts
    \[
        \cmax[f'] = c(f',3) = y'(p_3) \cdot \trtp[p_3] + y'(p_4)\cdot \trtp[p_4] = 3+3=6 > \cmax.
    \]
It is easy to see that both flows~$f$ and~$f'$ are maximum \TR flows of value~$6$. Thus, our claim holds.
\end{proof}

In contrast to the maximum flow over time problem, \TR solutions have arbitrarily bad peak costs compared to general flows over time on the same instance, i.e. they yield no constant approximation ratio.

\begin{lemma}\label{rem:TRFsubopt}
	Temporally repeated flows do not provide a constant factor approximation for the Minimum Peak Cost Flow problem.
\end{lemma}

\begin{proof}
Let $k\in \N$ be arbitrary but fixed.
	Consider a network on a graph $G = (V,A)$ with nodes~$V = \{s, v, w, t\}$ and arcs $A = \{(s,v),(v,t), (v,w), \allowbreak (w,t)\}$, shown in Fig.~\ref{fig:gapTR}, and a time horizon $\T \coloneqq 2k+2$.
All arcs have unit capacity; transit times and costs are as follows:
\[
\begin{array}{llll}
	\trt[(s,v)]=1, &\qquad \trt[(v,t)] = k, &\qquad \trt[(v,w)] = 1, &\qquad \trt[(w,t)] = k,\\
	\ca[(s,v)] = 0, &\qquad \ca[(v,t)] = 1, &\qquad \ca[(v,w)] = 0, &\qquad \ca[(w,t)] = 0.
\end{array}
\]
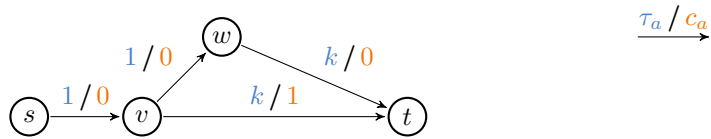
\begin{figure}[hb]
    \centering
    \input{tikz/TR-gap}
    \caption{An instance of MPC-MF  with unit capacities, for which the gap between the optimal peak cost and the peak cost of an optimal \TR flow is equal to $k \in \N$.}
    \label{fig:gapTR}
\end{figure}
We set the demand to the maximum flow value $\dem \coloneqq \dem^{\text{max}} = \T-k-1$.

Note that $G$ contains exactly two $s$-$t$ paths $p_1=(s,v,t)$ and $p_2=(s,v,w,t)$. 
Since $\trtp\leq T$ holds for both paths, 
\Cref{thm:FoFu} implies that a maximum \TR flow is induced by the unique static minimum-cost circulation, which has value $\val[x] = 1$ and uses only the shorter path $p_1$.
The unique maximum \TR flow $f^{\text{TR}}$ thus uses only the path~$p_1$ and sends flow at rate $\pflow[p_1]=1$ along path~$p_1$ in the time period $\big[0,\, \T - \trtp[p_1]\big) = [0,\, \T-k-1)$.
The maximum flow value is $\val[f^{\text{TR}}] = \T-k-1$, and the resulting flow rates on arcs are as follows:
\begin{align*}
    f^{\text{\textsc{TR}}}_{(s,v)}(\theta)=1 &\text{ for }\theta \in [0,\,\T-k-1),\\
    f^{\text{TR}}_{(v,t)}(\theta)=1 &\text{ for }\theta \in [1,\,\T-k),\\
\end{align*}
and zero otherwise.

Since only arc $(v,t)$ has a positive cost coefficient, the cost of the temporally repeated flow~$f^{\text{TR}}$ at a time point~$\tp$ is 
\[
    c(f^{\text{TR}}, \tp) = \ca[(v,t)] \cdot \int_{\tp - \trt[(v,t)]}^{\tp} f_{(v,t)}(\auxtp)\,d\auxtp = 
    1\cdot\int_{\tp - k}^{\tp} f_{(v,t)}(\auxtp)\,d\auxtp.
\]
{For instance, at time point $\tp=1$,  
no flow particles have reached arc $(v,t)$ yet, so the cost at this time point is zero.}
The peak cost of flow $f^{\text{TR}}$ is attained when the arc $(v,t)$ carries flow on its entire length, i.e.~at each time point between $k+1$ and $\T-k$; we calculate the cost at time point~$\tp = k+1$ and obtain
\[
\cmax[f^{\text{TR}}] =  \int_{1}^{k+1} f^{\text{TR}}_{(v,t)}(\auxtp)\,d\auxtp = k.
\]

Now consider a non \TR flow $f^*$, which sends the flow at rate $1$ over the longer but cheaper path $p_2$ in the time period $[0, \trtp[p_2]) = [0,\ \T-k-2)$. The last missing unit of flow is sent over the path $p_1$, departing in period $[\T-k-2,\ \T-k-1)$.
Formally, the flow $f^*$ is defined by the following flow rates on the arcs:
\begin{align*}
    &f^*_{(s,v)}(\theta)=1 \text{ for }\theta \in [0,\ \T-k-1),\\
    &f^*_{(v,w)}(\theta)=1 \text{ for }\theta \in [1,\ \T-k-1),\\
    &f^*_{(w,t)}(\theta) =1 \text{ for }\theta \in [2,\ \T-k),\\
    &f^*_{(v,t)}(\theta)=1 \text{ for }\theta \in [\T-k-1,\ \T-k).
\end{align*}
The flow rates outside of the given intervals are zero.
It is easy to see that $f^*$ is a feasible flow with $\val[f^*] = \T-k-1$; hence, flow~$f^*$ is also a maximum flow.
Flow $f^*$ also attains its \peakcosts when the amount of flow on arc $(v,t)$ is maximised, i.e.~at each time point $\tp \in [\T-k,\ \T-1)$. 
We compute the \peakcosts of flow~$f^\ast$ as cost at time point~$\tp=T-k$ and obtain
\[
\cmax[f^*] = \int_{\tp-k}^{\tp} f^*_{(v,t)}(\auxtp)\,d\auxtp = \int_{\T-k-1}^{\T-k}1\, d\auxtp = 1. 
\]
Hence, the ratio between the optimal \peakcosts of a \TR flow and of an unrestricted  optimal flow is at least~$\frac{\cmax[f^{\text{TR}}]}{\cmax[f^*]} = k$.
\end{proof}

Numerous applications of flows over time involve units of flow that are discrete by nature, for instance vehicles in traffic management or beds in a hospital. 
In these cases, we seek \emph{integral} flows over time, i.e.~flows with integral flow rates. 
For the maximum flow over time problem, the integrality constraint can be imposed without loss of generality: if arc capacities are integers, then there always exists an integral minimum cost static circulation, which then always yields an integral path decomposition and induces an integral maximum temporally repeated flow.
We lose this property when we consider the \minpeakcosts objective.

\begin{lemma}\label{rem:integer-gap}  
For \genProb, the \peakcosts of an optimal integral solution is, in the worst case,~$\Omega(\sqrt{n})$ times higher than the optimal \peakcosts, where $n$ is the number of nodes in the network.
\end{lemma}
\begin{proof}
    Consider the following network $\net = \netwDef$. 
    Source $s$ has one outgoing edge to node~$v$. 
    Node $v$ and target $t$ are connected by $k$ internally disjoint paths, each of length~$k$, for some integer~$k\in \mathbb{N}$; see 
    Fig.~\ref{fig:bsp-intGap}.
    \begin{figure}[hbt]
        \centering
        \resizebox{0.6\textwidth}{!}{
            \input{tikz/bsp-intGap}
        }
        \caption{A network for $k=4$ with unit transit times and capacities, for which any optimal integral solution for \genProb with~$\T=k+2$ and $\dem=1$ is by a factor ${k}$ more expensive than the optimal fractional solution. Arc costs that are not indicated are zero.}
        \label{fig:bsp-intGap}
    \end{figure}
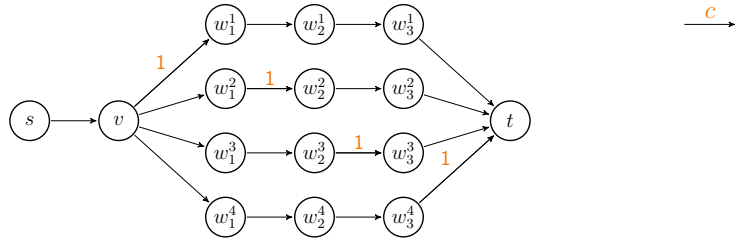
    Formally, we have 
    \newpage
    \begin{align*}
        V = &\{s,v,t\} \cup \{w^{i}_{j} \mid i\in \oneto{k},\ j \in \oneto{k-1}\},\\
        A = &\{(s,v) \} \cup \big\{ (w^{i}_{j-1},\ w^{i}_{j})\mid i\in \oneto{k},\ j\in\{2,\ldots,k-1\} \big\}\\
        &\phantom{\{(s,v) \}} \cup \{(v, w^{i}_{1}),\  (w^{i}_{k-1},\ t) \mid i \in \oneto{k}\}.
    \end{align*}
    Hence, the network contains $n = k(k-1) + 3$ nodes.
    
    All capacities and transit times are equal to one. 
    The cost of the $i$-th arc on the $i$-th $v$-$t$ path,~$i\in \oneto{k}$, equals one, and other arc costs are zero, i.e.
    \[
        \ca = \begin{cases}
            1, \text{ if } a= (w^{i}_{i-1},\ w^{i}_{i}) \text{ for }i\in\oneto{k},\\
            0, \text{ otherwise,}
        \end{cases}
    \]
    where we overwrite the notation via $w^{i}_{0} \coloneqq v$ and $w^{i}_{k} \coloneqq t$ to simplify the presentation. 
    Finally, we set the time horizon~$\T = k+2 $. 
    Since each $s$-$t$ path has length $k+1 = \T-1$, the flow can depart at node $s$ only within the time interval $[0,1)$. 
    The arc $(s,v)$ presents a capacity bottleneck and ensures that at most one unit of flow traverses the graph in the period $[0,1)$; hence, the maximum possible value of a flow over time is $1$, and we set $\dem=1$.
    
    Under the integrality constraint, the entire flow unit flows over one $v$-$t$~path $(v,\ w^{i}_{1},\ldots,w^{i}_{k-1},\ t)$ for some choice of $i\in \oneto{k}$. 
    Any integral maximum \TR flow $f^\text{int}$, which sends flow through the~$i$-th path, is given by the following flow rates:
    \[
    \begin{array}{ll}
        f^\text{int}_{(s,v)}(\theta)=1 &\text{ for }\theta \in [0, 1),\\[1ex]
        f^\text{int}_{(w^{i}_{j-1},\ w^{i}_{j})}(\theta) =1 &\text{ for }\theta \in [j, j+1) 
        \text{ for all }j\in \{1,\ldots,k\}
        ,\\[1ex]
        f^\text{int}_{a}(\tp) = 0 
        &\text{ otherwise.}
    \end{array}
    \]
    For any choice of index $i$, the flow incurs a \peakcosts of $\cmax[f^\text{int}] = 1$ when passing the unique arc with cost $1$ on the flow-carrying path:

    \[
        \cmax[f^\text{int}] = c(f^\text{int},\ i+1) = \int_{i}^{i+1} 1\cdot f_{(w^{i}_{i-1}, w_{i}^{i})}(\auxtp)\,d\auxtp = 1.
    \]
        
    Without the integrality constraint, we can distribute the load and send a flow $f^*$ at rate $\frac{1}{k}$ over each of the $k$ parallel paths. 
    Formally, 
     we define flow $f^*$ by the underlying flow decomposition 
    \[
        \pflowvar^* \colon \paths \to \R, \quad p\mapsto \frac{1}{k},
    \]
    which yields the following arc flow rates:
    \[
    \begin{array}{ll}
        f^*_{(s,v)}(\theta)=1 &\text{ for }\theta \in [0, 1),\\[1ex]
        f^*_{(w^{i}_{j-1},\ w^{i}_{j})} (\theta) =\frac{1}{k}&\text{ for }\theta \in [j, j+1) \text{ for all }j\in \{1,\ldots,k\},\ i\in \oneto{k}, \\[1ex]
        f^*_{a}(\tp) = 0 
        &\text{ otherwise.}
    \end{array}
    \]
    Observe that flow $f^*$ is still \TR.
    At any point in time, only a $\frac{1}{k}$-fraction of the flow traverses the arcs with nonzero costs; therefore, the peak cost of this fractional flow $f^*$ is 
    \[
        \cmax[f^*] = \max_{\tp} \sum_{i=1}^{k} \left(1 \cdot \int_{\tp-1}^{\tp} f_{(w^{i}_{i-1},\ w^{i}_{i})}(\auxtp) d\auxtp \right)
         = \max_{\tp} \sum_{i=1}^{k} \left(\frac{1}{k} \cdot \bigabs{[i,i+1) \cap [\tp-1, \tp)} \right) = \frac{1}{k}.
    \]
    Hence, since $n\sim k^2$, the ratio between the best objective value of an integral and a fractional solution is at least
    \[
        \frac{\cmax[f^{\text{int}}]}{\cmax[f^*]} = k \in \Omega(\sqrt{n})
    \] 
    in worst case.
\end{proof}

The above example shows that restricting the solution space to integral temporally repeated flows may lead to an arbitrarily large increase in peak costs; that is, general optimal solution are better.
However, the complexity and structure of general optimal solutions remain unknown.

%% file: tikz/path_decomp.tex


\begin{tikzpicture}[>=stealth',shorten >=1pt, shorten <=0.5pt, auto, node distance = 1.5cm]
    \tikzstyle{every state}=[thick, inner sep=0mm, minimum size=5mm]

  \node[state] (s)   {$s$};
  \node[state] (A) [above right of=s] {\small$v_1$};
  \node[state] (B) [below right of=A] {\small$v_2$};
  \node[state] (C) [above right of=B] {\small$v_3$};
  \node[state] (t) [below right of=C] {$t$};

  \path[->] (s) edge node[font=\small]{$1$} (B)
                edge node[font=\small]{$1$} (A)
            (A) edge node[font=\small]{$1$} (B)
            (B) edge node[font=\small]{$1$} (C)
                edge node[font=\small]{$1$}    (t)
            (C) edge node[font=\small]{$1$} (t);

    \node (x) [right=2.2 of C] {};
    \path[->, draw] (x) --node{\small$x$}  ++(0.9,0) {};
    
\end{tikzpicture}

%% file: tikz/TR-gap.tex

\begin{tikzpicture}[>=stealth',shorten >=1pt, shorten <=0.5pt, auto, node distance = 1.5cm]
    \tikzstyle{every state}=[thick, inner sep=0mm, minimum size=5mm]

  \node[state] (s)   {$s$};
  \node[state] (v) [right of=s] {$v$};
  \node[state] (w) [above right of=v] {$w$};
  \node[state] (t) [right=3 of v] {$t$};

  \path[->] (s) \tauuc{$1$}{}{$0$}{}{} (v)
            (v) \tauuc{$1$}{}{$0$}{}{} (w)
                \tauuc{$k$}{}{$1$}{}{} (t)
            (w) \tauuc{$k$}{}{$0$}{}{} (t);

    \node (x) [right=5 of w] {};
   \node (xx) [right=1.0 of x] {};
    \path[->] (x) \tauuc{$\trt$}{}{$\ca$}{}{} (xx);

\end{tikzpicture}

%% file: tikz/bsp-intGap.tex

\begin{tikzpicture}[>=stealth',shorten >=1pt, shorten <=0.5pt, auto, node distance = 1.5cm]
    \tikzstyle{every state}=[thick, inner sep=0.5mm, minimum size=8mm]
    \tikzstyle{every edge} = [draw, ->]

    \def\xstep{18mm}
    \def\ystep{13mm}
    \def\stepfrac{0.2}

  \node[state] (s) at (-1.2*\xstep, 0)  {\large$s$};
  \node[state] (v) at (-0.2*\xstep, 0) {\large$v$};
  \node[state] (t) at (4.2*\xstep, 0) {\large$t$};

    \foreach \y[count=\i] in {1.5, 0.5, -0.5, -1.5}{
        \foreach \j in {1,2,3}{
            \node[state] (w\i\j) at (\j*\xstep, \y*\ystep) {\large$w^{\i}_\j$};
        }
        \draw (v) edge (w\i1);
        \draw (w\i1) edge (w\i2);
        \draw (w\i2) edge (w\i3);
        \draw (w\i3) edge (t);
    }
    \draw[->] (s) -- (v);

    \foreach \a/\b in {v/w11, w21/w22, w32/w33, w43/t}{
        \path[->] (\a) \tauuc{}{}{{\large{1}}}{}{} (\b);
    }

    \node (x) [right=5 of w13] {};
    \path[->, draw] (x) --node{\Large\textcolor{cost}{$c$}}  ++(1.2,0) {};

\end{tikzpicture}

%% file: sections/reduction.tex
\section{Complexity of the integral \genProb}\label{sec:np-hard}
Similarly to the min-cost \TR flow problem, the integral \genProb and its decision counterpart are \np-hard already in a very restricted case.

\begin{theorem}\label{lem:int-strongNPhard}
    Given an instance of \genProb and  a number $z\in \Q_{+}$, it is \np-hard to decide whether there exists an integral \TR flow for the given instance with \peakcosts at most $z$, even for two-terminal series-parallel graphs with unit transit times, unit capacity and two possible arc cost values.
\end{theorem}
\begin{proof}
    We prove the statement by a reduction from \textsc{3-SAT} with a restriction that each clause contains exactly three pairwise different literals. 
    This restriction can be ensured by padding shorter clauses with dummy literals.
    
     Let $\inst$ be an instance of \textsc{3-SAT} with $n$ variables $X_i$, $i\in \oneto{n}$, and $m$ clauses $C_j$, $j\in \oneto{m}$.     
     We construct an instance $\iinst = (\net, \T, \dem)$
     of the decision version of \genProb similarly to the construction in \Cref{rem:integer-gap}.
     The network $\net=\netwDef$ is based on a graph~$G = (V,A)$ that contains a source $s$, a sink $t$, nodes $v_i$ connected to the sink for each variable $X_i$, $i\in \oneto{n}$, and a simple $s$-$v_i$ path 
     of length $m+2$ for each literal $X_i$ or $\overline{X_i}$; see also Fig.~\ref{fig:3sat}.
     We say that the $(j+1)$-th arc of every path \emph{corresponds} to clause $C_j$.
     
     Formally, we have     
     \begin{align*}
     V  = &\{s,t\} \cup \big\{v_i \mid i\in \oneto{n}\big\} \cup \big\{w_i^j, \overline{w}_i^j \mid i\in \oneto{n},\ j\in \zeroto{m} \big\}\text{ and}\\
     A = &\big\{(s, w_i^0), (s, \overline{w}_i^0), (w_i^m, v_i), (\overline{w}_i^m, v_i), (v_i,t) \mid i\in \oneto{n}\big\} \\
     \cup &\big\{(w_i^{j-1}, w_i^{j}), (\overline{w}_i^{j-1}, \overline{w}_i^{j}) \mid i \in \oneto{n},\ j\in \oneto{m}\big\}.
     \end{align*}

    All capacities and transit times are equal to one. 
    The arc costs are defined as follows:
    \[
        \ca[] \colon A\to \R_+,\quad a \mapsto \begin{cases}
            1, \text{ if } a=(w_i^{j-1}, w_i^{j}) \text{ and } \overline{X_i} \in C_j,\\
            1, \text{ if } a=(\overline{w}_i^{j-1}, \overline{w}_i^{j}) \text{ and } X_i \in C_j,\\
            0, \text{ otherwise};
        \end{cases}
    \]
    that is, for each literal $\ell$ in a clause $C$, the arc corresponding to clause $C$ in the path of the negated literal $\overline{\ell}$ has cost of one.
    This choice of costs later allows us to encode the number of negative literals in each clause by the cost at a corresponding time point.
    We set the time horizon \T to 
        $m+4$
    and ask for a 
    maximum \TR flow over time for this horizon, i.e.~for a flow of value $\dem=n$, with \peakcosts at most $z=2$.

    \begin{figure}[htb]
        \centering
        \resizebox{0.9\textwidth}{!}{
            \input{tikz/3SAT}
        }
        \caption{Left: an instance of \textsc{3-SAT} with $m=2$ clauses and $n=3$ variables.
        Right: the corresponding flow over time network.
        Transit times and capacitites are all equal to one;
        arc costs that are not explicitly indicated are equal to zero.}
        \label{fig:3sat}
    \end{figure}
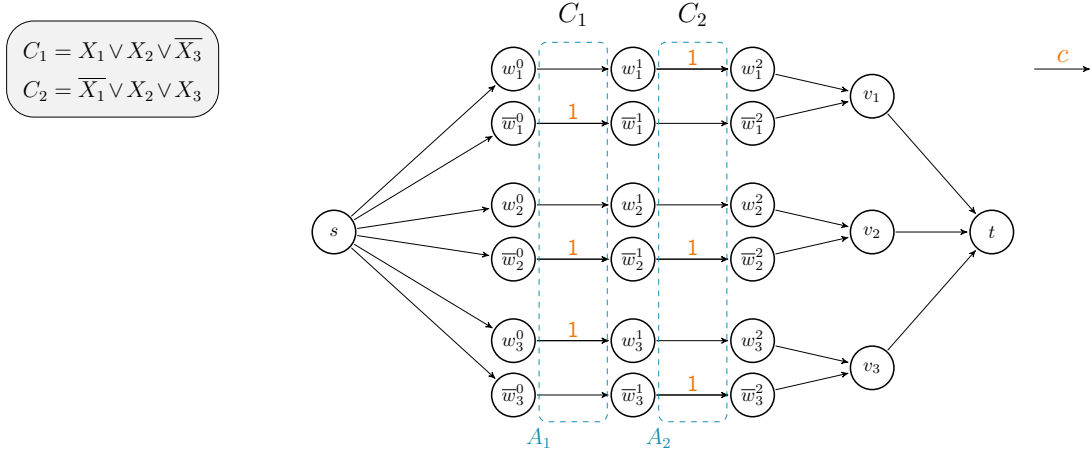

    We denote by \[
    A_j \coloneqq \big\{(w_i^{j-1}, w_i^j), (\overline{w}_i^{j-1}, \overline{w}_i^j) \mid i\in \oneto{n}\big\}
    \]
    the set of arcs \emph{corresponding} to the clause $C_j$, $j\in \oneto{m}$.
    For ease of notation, we analogously define arc sets 
    \[
    A_0 \coloneqq \big\{(s, w_i^0), (s, \overline{w}_i^0) \mid i\in \oneto{n}\big\} \quad\text{and}\quad
    A_{m+1}  \coloneqq \big\{(w_i^{m}, v_i), (\overline{w}_i^{m}, v_i) \mid i\in \oneto{n}\big\}.
    \]
    
    All $s$-$t$ paths in the graph have transit time $m+3 = \T-1$.
    To meet the time horizon, the flow is sent from the source only in the period $[0,1)$.
    Any maximum flow  has value $n$, which is dictated by the capacity available on arcs~$\{v_i,t\}$ with~$i\in [n]$.
    Any feasible, i.e.~integral maximum \TR flow $f$ in the constructed network thus has
    the following structure: it uses $n$ internally-disjoint paths $p_i$, $i\in \oneto{n}$, at full capacity, where $p_i$ is one of the paths
    \[
        q_i \coloneqq (s, w_i^0, \ldots, w_i^m, v_i, t) \qquad \text{or} \qquad \overline{q}_i \coloneqq (s, \overline{w}_i^0, \ldots, \overline{w}_i^m, v_i, t).
    \]
    Formally, the underlying path decomposition is $\pathdecomp$ with  $\pflow=1$ if and only if~$ p\in \{p_i\}_{i\in \oneto{n}}$ and~$\pflow=0$ otherwise.
        
        Hence, a feasible flow $f$ satisfies
        \[
            f_a(\tp) = \begin{cases}
                1,\quad&\text{if } a=(v,w)\in p_i \text{ for some }i\in\oneto{n} \text{ and }\tp \in \big[\trtp[p_i\vert_{s,v}],\, \trtp[p_i\vert_{s,v}]+1\big),\\
                0,\quad&\text{otherwise}.
            \end{cases}
        \]
        In particular, for an arc $a\in A_j$ belonging to a flow-carrying path we have $f_a(\tp) = \indic{[j,j+1)}(\tp)$
        \footnote{Function $\indic{S}\colon \R \to \{0,1\}$ for a set $S\subseteq \R$ is the indicator function with $\indic{S}(x) =1$ if and only if $x\in S$.},
        as $ j= \trtp[p_i\vert_{s,v}]$ is the length of each subpath up to an arc $a = (v,w)\in A_j$.
        
        Next, we compute the cost at each time point.
        The cost at each time point $\tp \notin [1,m+2)$ is zero.
        For~\mbox{$\tp \in [1,m+2)$}, the cost is a sum of costs of chain flows $f_i$ over paths $p_i$:
        \[
        c(f,\tp) = \sum_{i=1}^{n} c(f_i,\tp).
        \]
        For the chain flow $f_i$ over a path $p_i$, $i\in \oneto{n}$, departing in the period $[0,1)$, \instantcosts $\tp \in [1,m+2)$ is%
        \newcommand{\flotheta}{\flo{\tp}}%
        \begin{align*}
            c(f_i, \tp) 
            =& \sum_{a\in p_i} \ca \cdot \int_{\tp-\trt}^\tp f_{a}(\auxtp)\,d\auxtp
            = \sum_{a\in p_i} \ca \cdot \int_{\tp-1}^\tp f_{a}(\auxtp)\,d\auxtp
            \\\overset{(1)}{=}& \sum_{\substack{j\in\oneto{m} \cup \{0,m+1\}:
            \\[0.8ex] [j,j+1)\cap[\tp-1,\tp] \neq \emptyset}}
            \ca[]{(a_i^j)} \cdot \int_{\tp-1}^\tp f_{a_i^j}(\auxtp)\,d\auxtp
            \\\overset{(2)}{=}& \,\ca[](a^{\flotheta-1}_i) \cdot \int_{\tp-1}^\tp f_{a_i^{\flotheta-1}}(\auxtp)\,d\auxtp 
            \ +\ \ca[](a^{\flotheta}_i) \cdot \int_{\tp-1}^\tp f_{a_i^{\flotheta}}(\auxtp)\,d\auxtp
            \\\overset{(3)}{=}& \, \ca[](a^{\flotheta-1}_i) \cdot \int_{\tp-1}^\tp \indic{\big[\flotheta-1, \flotheta\big)}(\auxtp)\,d\auxtp 
            \ +\ \ca[](a^{\flotheta}_i) \cdot \int_{\tp-1}^\tp \indic{\big[\flotheta, \flotheta+1\big)}(\auxtp)\,d\auxtp
            \\\overset{}{=}& \,\ca[](a^{\flotheta-1}_i) \cdot (\flotheta - \tp +1) \ +\ \ca[](a^{\flotheta}_i) \cdot (\tp-\flotheta),
        \end{align*}
        where $a_i^j$ is the unique arc in $p_i\cap A_j$ for $j\in \zeroto{m+1}$.
        Equality~(1) 
        preserves  in the sum  only those arcs of path~$p_i$ that have non-zero flow rate in time period~$[\tp-1,\tp]$.
        Equality~(2) is true since $\lfloor\tp\rfloor-1$ and $\lfloor\tp\rfloor$ are exactly the two integers with $[j,j+1)\cap[\tp-1,\tp] \neq \emptyset$.
        It expresses the fact that, as the flow is sent for exactly one time unit, at most two incident arcs of $p_i$ carry flow and have an impact on the cost.
        For equality~(3), we substitute the expression for the flow rate.
        
        Equations above imply that the cost at an integer time point $\tp \in \N$ is
        \[
            c(f_i, \tp) = \ca[]{(a_i^{\tp-1})},
        \]
        while the cost at a fractional time point $\tp \notin \N$ is a convex combination of the costs of the two surrounding integer time points, and thus
        \[
            c(f_i, \tp) \leq \max\big\{c(f_i, \flotheta), c(f_{i}, \left\lceil{\tp}\right\rceil)\big\}.
        \]
        The entire flow thus also attains its peak cost at an integer time point:
        \[
        \cmax =
        \max_{\tp\in \{1,\ \ldots,\ m+2\}} \sum_{i=1}^{n} c(f_i, \tp)
                =        
        \max_{j\in \{0,\ \ldots,\ m+1\}} \sum_{i=1}^{n} \ca[](a_i^{j})
        =        
        \max_{j\in \{1,\ \ldots,\ m\}} \sum_{i=1}^{n} \ca[](a_i^{j}),
        \]
        as $\ca[](a_i^{0}) = \ca[](a_i^{m+1}) = 0$ for all $i\in \oneto{n}$.
        
    The considerations above hold for any integral \TR flow of value $n$ in the constructed instance~\iinst.
    Next, we show that instance \inst is a Yes-instance of 3-SAT if and only if \iinst is a Yes-instance of \genProb, i.e.~if it admits a flow with \peakcosts at most $2$.
    
    Let \inst be a Yes-instance, and let $\phi\colon \{X_i\}_{i=1}^n\to\{\True,\ \False\}$ be a satisfying truth assignment.
    Then each clause contains at least one literal with value \True.
    We construct a corresponding flow over time~$f$ for instance~\iinst as follows:
    for each $i\in \oneto{n}$ with $\phi(X_i) = \True$, send flow at rate one in time period $[0,1)$  over the path 
    \[
    p_i \coloneqq 
    q_i = (s,w_i^0,\ldots, w_i^m,v_i,t),
    \] 
    and for each $i\in \oneto{n}$ with $\phi(X_i) = \False$, send one unit of flow over the path 
    \[
    p_i \coloneqq 
    \overline{q_i} = (s, \overline{w}_i^0,\ldots, \overline{w}_i^m,v_i,t).
    \]
    As discussed above, the constructed flow $f$ is a feasible flow of value $n$ for time horizon $\T$.
    Its cost at any integer time point $\tp \in \{2,\ldots, m+1\}$ and for the corresponding clause number $j\coloneqq \tp-1$ is 
    \begin{align}
        c(f, \tp) 
        = 
        \sum_{i=1}^{n} c(f\vert_{p_i}, \tp) \tag{$*$} \label{eq:flow2clause}
        & \overset{(1)}{=}
        \sum_{i=1}^{n} \ca[](
        p_i\cap A_{j}
        )
        \\\nonumber
        & \overset{(2)}{=}
        \sum_{i\colon\!X_i\in C_{j}} \indic{\overline{q}_i}(p_i) + \sum_{i\colon\!\overline{X_i} \in C_j} \indic{q_i}(p_i)
        \\\nonumber
        & \overset{(3)}{=} \sum_{\substack{i\colon\;X_i\in C_{j}\\\phi(X_i)=\False}} 1 + \sum_{\substack{i\colon\;\overline{X_i} \in C_j\\ \phi(X_i)=\True}} 1\\\nonumber        
        & = \sum_{\substack{\ell\in C_{j}, \\ \phi({\ell}) = \False}} 1 
        \\\nonumber
        &\leq 2.
    \end{align}
    Equality~(1) holds since only arcs corresponding to the~$j$-th clause incur costs at time point~$\tp$. 
    Equality~(2) is true since only the occurrence of the positive literal~$X_i$ in clause~$C_j$ implies cost of one on path~$\overline{q}_i$, and only the occurrence of negative literal~$\overline{X}_i$ in clause~$C_j$ implies cost of one on path~$q_i$; these cost apply only if the said path coincides with $p_i$, which is denoted by the indicator functions.
    Equality~(3) uses the correspondence between the paths of the constructed flow $f$ and the truth assignment.
    The last inequality is true, since clause $C_{j}$ contains, by assumption, at least one literal $\ell$ with value $\phi(\ell) = \True$.
    Overall, since the cost at every integer time point is at most two, we obtain $\cmax \leq 2$.

    Now let the instance \iinst be a Yes-instance, and let $f$ be a flow with $\cmax \leq 2$.
    There is exactly one flow unit traversing either the path $q_i$ or the path $\overline{q}_i$ for each $i\in \oneto{n}$.
    We define a truth assignment $\phi$ of the variables in instance \inst as follows:
        $\phi(X_i) \coloneqq \True$ if and only if $f(q_i) > 0$.
    By equations \eqref{eq:flow2clause}, the number of unsatisfied literals in a clause $C_j$, $j\in \oneto{m}$, is
    \[
        \abs{\{\ell\in C_{j} \mid \phi({\ell}) = \False\}} 
        {=} c(f,\:j+1) \leq \cmax \leq 2. 
    \]
    Hence, $\phi$ is a satisfying assignment for instance \inst, and \inst is a Yes-instance.
\end{proof}

We conclude that optimising the \peakcosts over integer \TR flows is at least \np-hard as well.
\begin{theorem}\label{cor:int-strongNPhard}
    Finding an integral minimum-{peak-cost} \TR flow of a given value and for a given time horizon is at least strongly \np-hard, even for two-terminal series-parallel graphs with unit transit times, unit capacity and two possible arc cost values.
\end{theorem}

    The key mechanism of the reduction in the proof of \Cref{lem:int-strongNPhard} 
    is that each unit of flow corresponds to a binary variable: the flow unit decides for strictly one of the two alternative paths $q_i$ or $\overline{q}_i$, corresponding to the two literals of the variable. 
    
    Let us call flows in which each flow unit follows exactly one $s$-$t$ path \emph{unsplit}.
    In the reduction above, the unsplit property is ensured by the requirement for the flow to be temporally repeated and integral.

    The statement of \Cref{lem:int-strongNPhard} as well as the reduction construction hold analogously for the discrete time model.
    In the discrete model, the flow is partitioned into singleton units that move through the network as a whole and can depart only at given discrete time points.
    Since, by design of the network, the flow can depart at the source only at time point $0$, any feasible flow for the instance is temporally repeated; hence, we can relax this requirement on the sought flow.
    However, the reduction design still requires the flow to be unsplit.
    In the discrete time model, this property is ensured by the \emph{integrality} constraint alone. 
    We conclude: in the discrete time model, finding a \minpeakcosts integral flow is \np-hard already on series-parallel graphs with unit capacities and transit times. 
    
    In the following sections, we continue the complexity analysis on two special classes of instances that are not covered by the reduction in \Cref{lem:int-strongNPhard}.

%% file: tikz/3SAT.tex

\pgfdeclarelayer{bg}
\pgfsetlayers{bg, main}	
\begin{tikzpicture}[>=stealth',shorten >=1pt, shorten <=0.5pt, auto, node distance = 1.5cm]
    \tikzstyle{every state}=[thick, inner sep=0.5mm, minimum size=8mm]

    \def\xstep{22mm}
    \def\ystep{25mm}
    \def\stepfrac{0.2}

  \node[state] (s) at (-0.5*\xstep, 0) {$s$};
  \node[state] (t) at (5.*\xstep, 0) {$t$};

    \node[state] (w10) at (\xstep, 1.2*\ystep) {$w_1^0$};
    \node[state] (u10) at (\xstep, 0.8*\ystep) {$\overline{w}_1^0$};
    \node[state] (w20) at (\xstep, 0.2*\ystep) {$w_2^0$};
    \node[state] (u20) at (\xstep, -0.2*\ystep) {$\overline{w}_2^0$};
    \node[state] (w30) at (\xstep, -0.8*\ystep) {$w_3^0$};
    \node[state] (u30) at (\xstep,  -1.2*\ystep) {$\overline{w}_3^0$};

    \foreach \i/\y in {1/1, 2/0, 3/-1}{
        \foreach \j in {1,2}{
            \node[state] (w\i\j) at (\j*\xstep + \xstep, 0.2*\ystep +\y*\ystep) {$w_{\i}^\j$};
            \node[state] (u\i\j) at (\j*\xstep + \xstep, \y*\ystep - 0.2*\ystep) {$\overline{w}_{\i}^\j$};
        }
        \node[state] (v\i) at (4*\xstep, \y*\ystep) {$v_{\i}$};

        \draw[->] (s) -- (w\i0);
        \draw[->] (w\i0) -- (w\i1); 
        \draw[->] (w\i1) -- (w\i2);
        \draw[->] (w\i2) -- (v\i);
        \draw[->] (v\i) -- (t);
        
        \draw[->] (s) -- (u\i0);
        \draw[->] (u\i0) -- (u\i1);
        \draw[->] (u\i1) -- (u\i2);
        \draw[->] (u\i2) -- (v\i);
    }

    \foreach \a/\b in {w11/w12, u10/u11, u20/u21, u21/u22, w30/w31, u31/u32}{
        \path[->] (\a) \tauuc{}{}{{\large{1}}}{}{} (\b);
    }

  \node[] at ($0.5*(w10) + 0.5*(w11) + (0, 10mm)$) {\Large$C_1$};
  \node[] at ($0.5*(w11) + 0.5*(w12) + (0, 10mm)$) {\Large$C_2$};

    \node (x) [right=4.5 of w12] {};
    \path[->, draw] (x) --node{\Large\textcolor{cost}{$c$}}  ++(1.2,0) {};


\newcommand{\arcframecolor}{cyan!65!black}
\draw[dashed, \arcframecolor, rounded corners] ($(w10.north east) + (0.18,0.2)$) rectangle ($(u31.south west) + (-0.18, -0.2)$);
\node[text=\arcframecolor] at ($(u30.south east) + (0.18, -0.5)$) {\large$A_1$};
\draw[dashed, \arcframecolor, rounded corners]
($(w11.north east) + (0.18,0.2)$) rectangle ($(u32.south west) + (-0.18, -0.2)$);
\node[text=\arcframecolor] at ($(u31.south east) + (0.18, -0.5)$) {\large$A_2$};

    \node (sat) [draw, rounded corners=0.5cm, fill=gray!10!white, left=5 of w10, inner sep=3mm, text width=3.3cm
    ]    {\large{$C_1 = X_1 \lor X_2 \lor \overline{X_3}$\\[1.4ex]
    $C_2 = \overline{X_1} \lor X_2 \lor {X}_3$}};   
\end{tikzpicture}

%% file: sections/unitCostIntro.tex
\newcommand{\w}[1][p]{\omega_{#1}}
\newcommand{\upFr}[1][p]{u^{\text{fr}}(#1)}
\newcommand{\algoname}{MSSP\xspace}
\newcommand{\pp}{\ensuremath{p^{*}}\xspace}
\renewcommand{\aa}{\overline{a}}
\newcommand{\atil}{\tilde{a}}
\newcommand{\subp}{\ensuremath{\restr[\pp]{v,w}}}
\newcommand{\za}[1][a]{Z(#1)}
\newcommand{\ytarg}{Y^{*}}
\newcommand{\qhat}{\hat{q}}
\newcommand{\qq}{\hat{q}}

\section{Unit-cost networks}\label{sec:unitCosts}
Having seen that the problem is strongly \np-hard in general, we identify two polynomially solvable cases in this and the next section.
In the proof of \Cref{lem:int-strongNPhard}, the cost function used in the reduction has values in $\{0,1\}$. 
The proof transfers to any cost function with at least two different cost values.

Now we consider the complementary case of unit arc costs.
In this case, the cost of a flow at any time point is proportional to the amount of flow in the network.
This allows us to derive a more specific expression for the \peakcosts value.

\begin{lemma}\label{lem:unitCostNet}
    Let $f$ be a \TR flow with time horizon $\T$ on a network $(G, \ua[], \trt[], \ca[])$ with unit costs, i.e.~$c\equiv 1$. Then flow~$f$ attains \peakcosts at time $\hat\tp \coloneqq \lfloor\frac{\T}{2}\rfloor$, i.e.
    \[
        \cmax = \ca[](f, \left\lfloor\frac{\T}{2}\right\rfloor)
    \]
    and has \peakcosts
    \begin{equation*}\tag{$\star$}\label{eq:costForUnitCase}
    	\cmax = \ca[](f, \left\lfloor\frac{\T}{2}\right\rfloor) 
    	= \sum_{\substack{p\in \paths,\\\trtp\leq \frac{\T}{2}}} \pflow \cdot \trtp + \sum_{\substack{p\in \paths,\\\trtp > \frac{\T}{2}}} \pflow \cdot \big(\T-\trtp\big),
    \end{equation*}
    where \pathdecomp is the path decomposition of flow $f$.
\end{lemma}
\begin{proof}
    Let \pathdecomp be the underlying path decomposition of the flow $f$.
    Recall that flow $f$ is a sum of chain flows $f^{\T}_p$ for $p\in \paths$ with $\pflow>0$
    (see \Cref{def:chainFlow}).
    Since the arc costs are all equal to one, the flow's \instantcosts $\tp \in [0,\T]$ is equal to the \emph{amount of flow} present in the network at the considered time point $\tp$, denoted by $\text{val}(f,\tp)$. 
    We calculate the flow value for each time point and for each chain flow $f^{\T}_p$ separately.

    The chain flow $f^\T_p$ for $p\in \paths$ departs at node $s$ in time period $\big[0, \T-\trtp\big)$ and reaches the sink $t$ in period~\mbox{$[\trtp, \T)$}. 
    Flow that reaches the sink node 
    disappears from the network.
    Hence, the amount of flow in the network grows in the period $\big[0, \T-\trtp\big)$ and diminishes in the period $[\trtp, \T)$.
    If the transit time of a path is~\mbox{$\trtp \leq \frac{\T}{2}$}, and thus $\trtp \leq \T - \trtp$ holds, then 
    \[
    \text{val}(f_p^\T, \tp) = \pflow \cdot \begin{cases}
        \tp, &\text{if } \tp < \trtp,\\
        \trtp, &\text{if } \trtp \leq \tp \leq \T - \trtp,\\
        \T - \tp, &\text{if } \tp > \T - \trtp.
    \end{cases}
    \]
    Hence, the maximum amount of flow is contained in the network in period $\big[\trtp,\ \T-\trtp\big)$, and, in particular, at time $\hat{\tp}$.
    
    If the transit time of the path is $\trtp > \frac{\T}{2}$ and $\trtp > \T - \trtp$, then 
    the last unit of the flow departs form the source before the first unit arrives at the sink, and the amount of flow on the path is thus 

    \[
    \text{val}(f_p^\T, \tp) = \pflow \cdot \begin{cases}
        \tp, &\text{if } \tp < \T -\trtp,\\
        \T - \trtp, &\text{if } \T - \trtp \leq \tp \leq \trtp,\\
        \T - \tp, &\text{if } \tp > \trtp,
    \end{cases}
    \]
    which attains its maximum at $\tp = \hat\tp$.

    In total, the value of the \peakcosts is
	\[
        \cmax = \ca[](f, \left\lfloor\frac{\T}{2}\right\rfloor) 
        = \sum_{p\in \paths} \text{val}(f^\T_p, \left\lfloor\frac{\T}{2}\right\rfloor)
        = \sum_{\substack{p\in \paths,\\\trtp\leq \frac{\T}{2}}} \pflow \cdot \trtp + \sum_{\substack{p\in \paths,\\\trtp > \frac{\T}{2}}} \pflow \cdot \big(\T-\trtp\big).
	\]%
\end{proof}

The property shown in \Cref{lem:unitCostNet} will be central in this section.
Therefore, we introduce the some auxiliary notation.

For a path~$p \in \paths$ and time horizon~$\T$ we define the \emph{weight}~$\w$ as 
\[	\w = \begin{cases}
	\trtp,&\quad \text{if } \trtp \leq \frac{\T}{2},\\
	\T-\trtp, &\quad \text{otherwise}. 
\end{cases}
\]
The weight of a path is the factor with which its chain flow contributes to the \peakcosts in networks with unit costs; that is, we restate expression~\eqref{eq:costForUnitCase} as
\[ \cmax = \sum_{p\in \paths} \pflow\w.
\]
Note further that $\w  = \min\{\trtp,\ \T - \trtp\}\leq \T - \trtp$ for any path $p$.

Equation~\eqref{eq:costForUnitCase} suggests that there is a structural difference between paths of length below and above the half of the time horizon $\frac{\T}{2}$.
In general, networks contain $s$-$t$ paths of both types.
The following special cases are therefore of particular interest.

\begin{definition}[Short and long time horizon]
	For an instance of \genProb with network $\net = \netwDef$ and time horizon~$\T$, we say that 
	the time horizon is \emph{short} if for any $s$-$t$~path~$p$ in~$\net$ we have $\trtp > \frac{\T}{2}$.
	Analogously, we say that \T is \emph{long} if we have $\trtp \leq \frac{\T}{2}$ for all $s$-$t$~paths~$p$.
\end{definition}

\begin{remark}\label{rem:unit-short-dem-val}
	Consider an instance of \genProb with network $\net = \netwDef$ with unit cost and a short time horizon~$\T$.
	Then, any chain flow $f^\T_p$ along path $p$ with flow rate $\pflow$ has flow value $\val[f^{\T}_{p}] = \pflow \cdot (\T-\trtp)$.
	For a \TR flow $f$ in $\net$ with path decomposition $\pathdecomp$ we have 
	\[	\val = \sum_{p\in \paths} \pflow (\T-\trtp) \qquad\text{ and }\qquad \cmax = \sum_{p\in \paths} \pflow \w =  \sum_{p\in \paths} \pflow (\T-\trtp) = \val
	\]
	by~\Cref{lem:unitCostNet}.
	That is, any \TR flow satisfying the given demand \dem has the same cost, and finding one feasible solution with demand of exactly \dem suffices to solve the optimisation problem.
\end{remark}
A feasible fractional solution of value exactly \dem is obtained from a maximum \TR flow by iteratively decreasing the flow rates along the paths of the path decomposition.
Since the Ford-Fulkerson algorithm for maximum \TR flows yields a path decomposition with polynomially many flow-carrying paths, reducing the flow until its value equals \dem is also possible in polynomial time.

Hence, we note the following result.
\begin{corollary}\label{lem:CGforShort}
	Fractional	\genProb is polynomial-time solvable on unit-cost networks if the time horizon is short.
\end{corollary}

In the integral case, greedily decreasing the flow rates of a path decomposition does not necessarily lead to a flow of value exactly \dem, if such a flow exists at all. 
In fact, even if such a flow exists, it is \np-hard to find a corresponding integral path decomposition, as \Cref{thm:SSP-reduct} demonstrates.

\begin{theorem}\label{thm:SSP-reduct}
	Integral \genProb with unit cost is at least weakly \np-hard, even for a short time horizon and on series-parallel graphs.
\end{theorem}
\begin{proof}
	We show the statement by reducing positive \SubSum to \genProb.
	An instance of positive \SubSum is given by a (multi-)set $\{x_1,\ldots, x_n\}$ of $n$  positive integers and a target sum $L\in \N$.
	The problem asks for a subset $S\subseteq \oneto{n}$ of indices such that the corresponding sum of integers is $\sum_{i\in S}x_i = L$.
	
	Given an instance $\inst$ of \SubSum, we construct a corresponding instance of \genProb as follows.
	The graph $G= (V,A)$ consists of nodes $s$ and $t$ and of $n$ distinct arcs $a_i$, $i\in \oneto{n}$, connecting $s$ and $t$.
	We set the cost and the capacity of each arc to one.
	Before we define transit times, we first set the time horizon to~$\T = 2\cdot \max_{i\in \oneto{n}}{x_i}$, and the demand to $\dem\coloneqq L$.
	The transit times now are $\trt[a_i] \coloneqq \T - x_i$ for each $i\in \oneto{n}$.
	
	We show that there exists a feasible integral solution for the \genProb instance with \peakcosts at most $L$ if and only if the instance \inst of \SubSum is a Yes-instance.
	First, let $\inst$ be a Yes-instance, and $S\subseteq\oneto{n}$ be a feasible solution.
	We construct a temporally repeated flow that uses exactly those arcs that correspond to indices in the subset $S$.
	Formally, we identify the set of paths $\paths$ with the arc set $A$ and consider the path decomposition~\mbox{$y\colon A \to \{0,1\}$} with $y(a) = 1$ if and only if $a = a_i$ for an $i\in S$.
	The resulting \TR flow $f$ has value 
	\[\val = \sum_{i\in \oneto{n}} (\T-\trtp[a_i])\cdot y(a_i) = \sum_{i\in S} (\T - (\T - x_i)) = \sum_{i\in S} x_i = L = D.
	\]
	The time horizon was chosen to be short, since~$\trtp[a_i] \geq \T - \max_{i\in \oneto{n}} x_i = \frac{\T}{2}$.
	As the arc costs are unit, by \Cref{rem:unit-short-dem-val} the \peakcosts of flow $f$ is
	\[\cmax = \sum_{i\in \oneto{n}} (\T-\trtp[a_i])y(a_i) = \val[f] = D.
	\]
	Hence, flow $f$ is the requested flow for the \genProb instance.
	
	For the other direction, let $f$ be a feasible integral \TR flow with value $\val \geq D$ and \peakcosts $\cmax \leq D$.
	Let $y\colon A \to \Z_{+}$ be its path decomposition.
	By \Cref{rem:unit-short-dem-val} we have $\val = \cmax = D$ and the flow value is 
	\[\val = \sum_{i\in \oneto{n}} \big(\T-\trtp[a_i]\big) y(a_i),
	\]  
	where $y(a_i) \in \{0,1\}$ due to arc capacities and flow's integrality.
	Hence, \[
	\sum_{i\in \oneto{n}\colon y(a_i)=1} \big(\T-\trtp[a_i]\big) = 	\sum_{i\in \oneto{n}\colon y(a_i)=1} x_i = D = L,
	\]
	so the index set $S\coloneqq\{ i\in \oneto{n}\mid y(a_i)=1\}$ is a solution for the \SubSum instance.
	
	Hence, integral \genProb is at least weakly \np-hard on series-parallel networks with unit costs.
\end{proof}

Series-parallel graphs tend to render flow problems easier to solve. 
Nevertheless, for the sake of completeness, note that a reduction similar to the proof of~\Cref{thm:SSP-reduct} exists for non-series-parallel graphs as well.

Since integral \genProb is \np-hard on unit-cost networks with a short time horizon, the same complexity results holds for general time horizons, i.e.~time horizons $\T$ such that the network contains paths with transit time both smaller and greater than $\frac{\T}{2}$.

%% file: sections/earlArrFlows.tex
Next, we show that, in contrast,  a \emph{maximum} \TR flow with minimum \peakcosts, i.e.~if a flow satisfying a demand equal to the maximum possible flow value, can be computed in polynomial time for unit-cost series-parallel networks.
To this end, we establish a link between \minpeakcosts maximum flows and earliest arrival flows.

\newcommand{\arr}[2][f]{\text{arr}_{#1}(#2)}
For a flow over time $f$ and a time point $\tp\geq 0$, let $\arr{\tp}$ denote the amount of flow that has reached the sink by time $\tp$.
An \emph{earliest arrival flow} $f$ is a feasible flow over time with the following property: the amount of flow $\arr{\tp}$ arrived at the sink by time $\tp$ is maximal for all $\tp \in [0,\T]$ simultaneously.
Clearly, earliest arrival flows are maximum flows.

For a \TR flow $f$ with a path decomposition $\pathdecomp$, the flow amount that reached the sink by time $\tp$ is 
\[
\arr{\tp} = \sum_{p\in \paths} \pflow \cdot \max\{\tp - \trtp,\ 0\} = \sum_{\substack{p\in \paths\\\trtp \leq \tp}} \pflow\cdot \big(\tp-\trtp\big).
\]

Next, we consider expression \eqref{eq:costForUnitCase} for the \peakcosts of a \TR flow in the case of unit costs and transform it as follows:
\begin{align*}
	\cmax =& \sum_{\substack{p\in \paths\\\trtp\leq \frac{\T}{2}}} \pflow \cdot \trtp + \sum_{\substack{p\in \paths\\\trtp > \frac{\T}{2}}} \pflow \cdot \big(\T-\trtp\big)\\
	=& \sum_{\substack{p\in \paths\\\trtp\leq \frac{\T}{2}}} \pflow \cdot \big(2\trtp - \T + \T - \trtp\big) +
	\sum_{\substack{p\in \paths\\\trtp > \frac{\T}{2}}} \pflow \cdot \big(\T-\trtp\big)\\
	=& \sum_{\substack{p\in \paths\\\trtp\leq \frac{\T}{2}}} \pflow \cdot (2\trtp - \T) + 
	\sum_{{p\in \paths}} \pflow \cdot \big(\T-\trtp\big)\\
	=& \sum_{\substack{p\in \paths\\\trtp\leq \frac{\T}{2}}} \pflow \cdot (2\trtp - \T) + \val\\
	=& -2 \sum_{\substack{p\in \paths\\\trtp\leq \frac{\T}{2}}} \pflow \cdot \left( \frac{\T}{2} - \trtp\right) + \val\\
	=&\, \val - 2\cdot\arr{\frac{\T}{2}}.
\end{align*}
Hence, for maximum \TR flows on unit-cost networks, minimising the \peakcosts is equivalent to maximising the amount of flow reaching the sink by time $\frac{\T}{2}$. 
Consequently, earliest arrival \TR flows have the smallest \peakcosts among maximum \TR flows.

Earliest arrival flows always exist in a network with a single source and a single sink \cite{Gale1959, Philpott90}; 
however, in general, \TR flows do not have the earliest arrival property.
Series-parallel graphs present an exception: 
Ruzika et al.~\cite{RuzikaSS11} show existence of a \TR flow which is an earliest arrival flow; moreover, this flow is found by a greedy polynomial-time algorithm. 
The algorithm is a variant of the successive shortest path algorithm by Bein et al.~\cite{BeinBT85} and builds the solution iteratively starting with an empty flow. 
In each step, it finds an $s$-$t$ path $p$ with the shortest transit time. 
As long as the path's transit time is smaller than the time horizon~$\T$, the algorithm adds a chain flow along path $p$ with a flow rate equal to its bottleneck capacity~$\ua[](p) \coloneqq \min_{a\in A} \ua$ to the solution.
Capacities of all arcs of path $p$ are then reduced by $\ua[](p)$. 
The algorithm stops once there are no~$s$-$t$ paths shorter than $\T$.

If all capacities are integral, then so are the resulting flow values. This final observation leads to the following result.

\begin{theorem}\label{thm:unit-cost-maxflow}
	An integral \minpeakcosts maximum \TR flow in a unit-cost series-parallel network can be found in strongly polynomial time. 
\end{theorem}
The described relation between \genProb and earliest arrival flows aligns well with the result of Fleischer and Skutella on minimum-cost flows: they show that in networks with unit costs, the universally quickest flow, i.e.~a flow with both the earliest-arrival and latest-departure property, has minimum cost \cite{FS2003}.

The greedy algorithm described above constructs an integral solution automatically, as long as the arc capacities are integral. 
In other words, on unit-cost series-parallel networks, \genProb with maximum demand has optimal integral solutions.

\begin{remark}
	On non-series-parallel graphs, the greedy algorithm does not even find a maximum flow.
	Consider the minimal not series-parallel graph $G$ show in \Cref{fig:counterex-greedy}.
	We set the costs and capacity of each arc to one; transit times are displayed in the figure.
	We set the time horizon to $\T=2k+2$ and demand to the maximum flow value~$\dem=2k+2$.
	
\begin{figure}
	\centering
	\resizebox{0.25\textwidth}{!}{
		\input{tikz/counterex-greedy}
	}
	\caption{A network on which the Successive Shortest Path algorithm cannot find a flow of maximum value $2k+2$ for $k\geq 3$. Capacities and costs are unit, transit times are depicted.}\label{fig:counterex-greedy}
\end{figure}
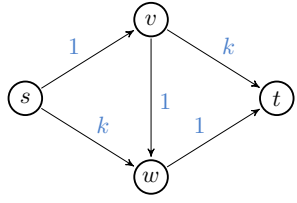
	
	The greedy algorithm chooses the unique shortest path $p = (s,v,w,t)$, which can transport at most \mbox{$1\cdot(\T - \trtp) = 2k-1$} units of flow.
	After path $p$ is used at full capacity, the algorithm can select no further path. 
	The total value of the flow produced by the greedy algorithm is $\val[f^{\textsc{alg}}] = 2k-1$.
	However, sending flow along the two disjoint paths $p_1 = (s,v,t)$ and $p_2 = (s,w,t)$ at rate $\pflow[p_1] = \pflow[p_2] = 1$ 
	yields a flow with total value~\mbox{$2\cdot 1\cdot (\T - \trtp[p_1]) = 2\cdot(k+1) = 2k+2 > \val[f^{\textsc{alg}}]$}.
\end{remark}

%% file: tikz/counterex-greedy.tex
	
	\begin{tikzpicture}[>=stealth',shorten >=1pt, shorten <=0.5pt, auto, node distance = 1.5cm]
		\tikzstyle{every state}=[thick, inner sep=0mm, minimum size=5mm]
		
		\node[state] (s)   {$s$};
		\node[state] (v) [above right= 0.8 and 1.5 of s] {$v$};
		\node[state] (w) [below right = 0.8 and 1.5 of s] {$w$};
		\node[state] (t) [below right = 0.8 and 1.5 of v] {$t$};
		
		\path[->] (s) \tauuc{$1$}{}{}{}{} (v)
			\tauuc{$k$}{}{}{}{} (w)
		(v) \tauuc{$1$}{}{}{}{} (w)
		\tauuc{$k$}{}{}{}{} (t)
		(w) \tauuc{$1$}{}{}{}{} (t);
		
	\end{tikzpicture}

%% file: sections/ssp-algo.tex
\renewcommand{\xa}[1][a]{Y(#1)}

In the remainder of the section, we consider the fractional version of the problem with unit-cost networks.
We show that a variant of the Successive Shortest Path algorithm yields a strongly polynomial algorithm for series-parallel networks for any demand values.
Moreover, we show that the fractional problem is polynomial-time solvable on any acyclic unit-cost networks.

\subsection{Successive Shortest Path Algorithm for Series-Parallel Graphs}
We show that the greedy Successive Shortest Path Algorithm as in \Cref{thm:unit-cost-maxflow} solves \genProb on series-parallel graphs for any demand value.

For a path $p$, we denote with 
$\ua[p]$ the bottleneck capacity  
\[ \ua[p] \coloneqq \min_{a\in p} \ua,
\]
of the path and, for given a path decomposition $\pathdecomp$ and a subset of arcs $S\subseteq p$, we denote with 
$\upFr[S]$ the remaining free capacity on the arcs of the path, defined as 
$$\upFr[S] \coloneqq \min_{a\in S}\big\{\ua - \sum_{p'\in \paths,\, a\in p'} \pflow \big\},$$
with the standard convention $\upFr[\emptyset] \coloneqq \min_{a\in \emptyset} = \infty$.

The modified greedy algorithm \algoname works as follows: 
we start with an empty flow~$f$.
In each iteration, let~$f$ be the current solution. 
Consider a shortest $s$-$t$~path with free capacity~$\ua[p] > 0$.
If there is no such path, stop: the instance is infeasible.
Otherwise, add a chain flow along~$p$ with rate $\pflow \coloneqq \min\{ \upFr,\, \frac{\dem - \val}{\T-\trtp}\}$.
Iterate until~$\val = \dem$.

We prove that the \algoname algorithm is correct similarly to the reasoning by Bein et al.\cite{BeinBT85}~for static minimum-cost flows. 
We first show that for a shortest path $p$ in the network, there always exists an optimal solution that uses this path at full capacity. 
Afterwards, we show by induction that the algorithm iteratively constructs an optimal solution.

In the following, we write~$p\circ q$ to denote a \emph{concatenation} of simple $s$-$v$ path~$p$ and a simple $v$-$t$ path~$q$. that are disjoint, to an $s$-$t$ path. 
For two paths $p$ and $q$, we denote with $p\setminus q$ the set of arcs $A(p)\setminus A(q)$.
Furthermore, we denote static flows by capital letters (e.g.~$Y\colon A\to \Q$) in contrast to path decompositions of \TR flows, denoted by small letters.
For a path decomposition \pathdecomp, we denote with $\paths_y^{+}$ the set of \emph{flow-carrying paths}, i.e.~paths $p\in \paths$ with $\pflow >0$. 
Finally, we will use notation $(f,y)$ to denote a flow $f$ with its path decomposition $\pathdecomp$.

\begin{lemma}\label{thm:shortestPathFull}
	Let \net be a series-parallel network and \pp a shortest $s$-$t$ path with respect to transit time.
	Then, for demand $\dem \leq D^\text{max}$, there always exists a \minpeakcosts \TR flow of value \dem that \emph{uses path \pp at full capacity}, i.e.~the path decomposition $\pathdecomp$ and the corresponding static flow $Y \colon A\to\Q$ satisfy
	\[ \begin{array}{ll}
		&\xa \geq \ua[\pp], \quad\text{ if } \dem \geq \ua[\pp](\T-\trtp[\pp]), \text{ or}\\[1ex]
		&\xa = \frac{\dem}{\T-\trtp[\pp]}, \text{ otherwise},
	\end{array} 
	\] 
	for each arc $a\in \pp$.
\end{lemma}

\begin{proof}
	Let 
	\[
	\ytarg \coloneqq \begin{cases}
		\ua[\pp], \text{ if } \dem \geq \ua[\pp](\T-\trtp[\pp]),\\
		\frac{\dem}{\T-\trtp[\pp]}, \text{ otherwise},
	\end{cases}
	\]
	be the desired value of the static flow on path $\pp$.
	
	We prove \Cref{thm:shortestPathFull} by contraposition.
	Let $f$ be a \minpeakcosts \TR flow of value \dem, and let $\pathdecomp$ be its path decomposition; 
	among all optimal flows, let $f$ be the one with the maximum total value of the static flow on path \pp,
	i.e.~flow~$f$ maximises the value
	\[\sum_{a\in \pp} \xa = \sum_{a\in \pp} \sum_{p\in \paths\colon a\in p} \pflow.
	\]
	
	We assume that none of the optimal flows uses path \pp at full capacity; 
	in particular, for flow~$f$ there are arcs~$a\in \pp$ with~$\xa < \ytarg$.
	Denote the set of such \emph{unsaturated} arcs by $U\coloneqq \{a\in \pp \mid \xa<\ytarg\}$, and the set of \emph{saturated} arcs by 
	$\overline{U} \coloneqq \{a\in \pp \mid \xa \geq \ytarg\}$.
	Note that our notion of saturation is coupled to the bottleneck capacity of the path \pp, that is, saturated arcs can still have slack capacity.
	
	Next, we define an auxiliary \emph{rerouting} operation and prove that if a rerouting to path \pp exists, then flow~$f$ could not have been optimal.
	\begin{claim}\label{cl:reroute}
		Let $q\in \pospaths$ be a flow-carrying path different from $\pp$, and suppose there exists a $\delta>0$ such that~$\delta \leq \pflow[q]$ and $\delta \leq \upFr[\pp\setminus q]$.
		Then a flow $f'$ defined by a path decomposition 
		\[
		z\colon \paths \to\Q,\qquad z(p) = \begin{cases}
			\pflow[q]-\delta, &\quad p=q,\\
			\pflow[\pp] + \delta\cdot\frac{\T-\trtp[q]}{\T-\trtp[\pp]}, &\quad p=\pp\\
			\pflow, &\quad \text{otherwise}.
		\end{cases}
		\] 
		is a feasible flow of value $D$ with either smaller \peakcosts or same \peakcosts and greater amount of total flow on path $\pp$. 
		
		We say that flow $f'$ is obtained by \emph{rerouting the flow} from path $q$ to path $\pp$.
	\end{claim} 
	\begin{claimproof}
		First, we show that the constructed flow $f'$ is feasible.
		Since $\delta \leq \pflow[q]$, all path flow rates remain non-negative. 
		Flow $f'$ also respects capacities of all arcs:
		for any arc $a\in A$, the total static flow $Z$ corresponding to the flow $f'$ has value
		\[	\za = \xa + \indic{\pp}(a)\footnote{Function $\indic{S}\colon \R \to \{0,1\}$ for a set $S\subseteq \R$ is the indicator function with $\indic{S}(x) =1$ if and only if $x\in S$.} 
		\cdot\delta\cdot\frac{\T-\trtp[q]}{\T-\trtp[\pp]} - \indic{q}(a) \cdot\delta =
		\begin{cases}
			\xa - \delta, &\text{ if } a \in q\setminus \pp,\\
			\xa + \delta\cdot\frac{\T-\trtp[q]}{\T-\trtp[\pp]}, &\text{ if } a\in \pp\setminus q,\\
			\xa -\delta(1 - \frac{\T-\trtp[q]}{\T-\trtp[\pp]}), &\text{ if } a\in \pp\cap q,\\
			\xa, &\text{ if } a \notin q\cup\pp.
		\end{cases}
		\]
		Hence, the greatest increase in the total flow value occurs for arcs $a\in \pp\setminus q$, for which we have
		$$\za \leq \xa + \delta\cdot\frac{\T-\trtp[q]}{\T-\trtp[\pp]} \overset{(1)}{\leq} \xa + \upFr[\pp\setminus q] \leq \ua,$$
		where inequality (1) holds since $\trtp[q] \geq \trtp[\pp]$ and since $\delta \leq \upFr[\pp\setminus q]$.

		For the value of the constructed flow over time we have 
		\begin{align*}
			\val[f'] &= \val + \big(z(\pp) - \pflow[\pp]\big)\big(\T - \trtp[\pp]\big) +\big(z(q) - \pflow[q]\big)\big(\T-\trtp[q]\big) 
			\\&=\val + \delta\cdot\frac{\T-\trtp[q]}{\T-\trtp[\pp]} \cdot \big(\T - \trtp[\pp]\big) - \delta \big(\T-\trtp[q]\big)
			\\&=\val = \dem,
		\end{align*}
		hence $f'$ is also a \TR flow of value \dem.
		
		Finally, we calculate the \peakcosts of the flow $f'$.
		If path \pp has transit time $\trtp[\pp] \geq \frac{\T}{2}$, then, by \Cref{rem:unit-short-dem-val}, any \TR flow of value \dem has the same peak cost, i.e.~$\cmax[f'] = \cmax$.
		Similarly, if paths \pp and $q$ have equal transit times, i.e.~$\trtp[\pp] = \trtp[q]$ and $\w[\pp] = \w[q]$,
		then  
		\begin{align*}
			\cmax[f'] = &\cmax + \big(z(\pp) - \pflow[\pp]\big)\cdot\w[\pp] +\big(z(q) - \pflow[q]\big)\cdot\w[q] 
			\\
			= &\cmax + \w[\pp] \cdot \delta \cdot \left(\frac{\T-\trtp[p']}{\T-\trtp[\pp]} - 1\right)
			\\ = & \cmax,
		\end{align*}
		so $f'$ is also a minimum-peak-cost flow. 
		In both cases, the amount of the static flow on path $\pp$ is
		\[\sum_{a\in\pp} \za  = \sum_{a\in \pp} \xa + \sum_{a\in \pp\setminus q} \delta \overset{\delta>0}{>} \sum_{a\in \pp}\xa; 
		\]
		that is, flow $f$ does not have the maximum total value of the static flow on path \pp.
		
		It remains to consider the case where the chosen path $q$ is longer than \pp, i.e.~$\trtp[q] > \trtp[\pp]$, and path \pp has length $\trtp[\pp] < \frac{\T}{2}$, and thus $\w[\pp] = \trtp[\pp]$.
		Then for the \peakcosts we have
		\begin{align*}
			\cmax[f'] = & \cmax + \big(z(\pp) - \pflow[\pp]\big)\cdot\w[\pp] +\big(z(q) - \pflow[q]\big)\cdot\w[q] 
			\\= & \cmax + \delta\cdot\frac{\T-\trtp[q]}{\T-\trtp[\pp]} \cdot \trtp[\pp] - \delta \cdot \w[q]
			\\<&\cmax + \delta \big(\trtp[\pp] - \w[q]\big) 
			\\\overset{\w[q] \leq \trtp[q] }{\leq} &\cmax + \delta \big(\trtp[\pp] - \trtp[q]\big)
			\\< &\cmax.
		\end{align*}
		Hence, rerouted flow $f'$ has smaller \peakcosts than flow $f$.
	\end{claimproof}

	Note that the existence of such a rerouting contradicts in both cases the choice of path \pp.
	
	Next, we consider the saturated arcs on path \pp. 
	We differentiate three cases by the number of such arcs and find a rerouting as in Claim~1 for each case, thus showing a contradiction to the existence of flow $f$.
	
	Recall that we assume $U\neq \emptyset$. 
	
	\noindent\underline{Case 1}:  $\overline{U} = \emptyset$. 
	Then the amount of flow transported over path \pp is strictly less than demand \dem, as otherwise we would have $\pflow[\pp] = Y^*$ and all arcs of \pp would be saturated.
	Hence, there exists another path $q\in \paths$ with~$\pflow[q] > 0$.
	
	As all arcs of path \pp have slack capacity, we can {reroute} the flow from path $q$ to path \pp: 
	set $$\delta \coloneqq \min\big\{\pflow[q],\ \upFr[\pp\setminus q] \big\},$$ 
	where $\pp \setminus q$ is considered as a set of arcs. 
	Note that  $\upFr[\pp] > 0$ implies $\delta > 0$.
	Since this choice of $\delta$ satisfies the condition of~\Cref{cl:reroute}, we obtain by rerouting  a contradiction to the choice of flow $f$.
	
	\vspace{1.5ex}\noindent\underline{Case~2}:  $\abs{\overline{U}} = 1$. 
	Let $a' \in \pp$ be the unique saturated arc.
	Since, by assumption, path \pp contains at least one unsaturated arc $\hat{a}$, there exists a flow-carrying path $q\in \paths$ with $a'\in q$ which is different from $\pp$:
	otherwise, if the only flow-carrying path containing arc $a'$ were \pp, then the static flow on arc ${a}$ would be
	\[
	\xa[{a}] = \pflow[\pp] + \sum_{p\in \paths,\, p\neq \pp} \pflow[p] \geq \pflow[\pp] = \xa[a'],
	\]
	which is a contradiction to ${a}$ being unsaturated.
	
	Analogously to Case~1, we can reroute a part of the flow with flow rate
	\[ \delta \coloneqq \min\big\{\pflow[q],\ \upFr[\pp\setminus q]\big\}
	\]
	from path $q$ to path \pp.
	Since the arc set $\pp\setminus q$ does not contain the unique saturated arc $a'$, alls arcs in this set are unsaturated. 
	So the slack capacity $\upFr[\pp\setminus q]$ is positive, and thus $\delta > 0$.
	
	The flow resulting from rerouting is feasible by~\Cref{cl:reroute}.
	Hence, we obtain again a contradiction to the choice of path \pp.
	
	\vspace{1.5ex}\noindent\underline{Case 3}:  $\abs{\overline{U}} \geq 2$.
	Let $v$ nd $w$ be nodes on path \pp such that the subpath $\subp$ is a maximum-length subpath consisting of unsaturated arcs.
	By our assumption, this subpath is not empty, i.e.~$v\neq w$. 
	
	Suppose there exists a path $q\in \pospaths$ with $q\neq \pp$ that contains both nodes~$v$ and $w$. 
	Then the path is directed from~$v$ to~$w$, as otherwise the subpath $\restr[q]{w,v}$ together with path \pp would close a cycle, which is forbidden in series-parallel graphs. 

	We show that path $q$ contains the entire subpath between $v$ and $w$ by contradiction.
	So let $\restr[q]{v,w} \neq \subp$.
	Since \pp is a shortest path, we have $\trtp[{\restr[q]{v,w}}] \geq \trtp[{\subp}]$.	
	Now consider a path $q' \coloneqq \restr[q]{s,v}\circ\subp\circ \restr[q]{w,t}$.
	Note that path~$q'$ is simple: if not, i.e.~if, without loss of generality, the subpath $\restr[q]{s,v}$ contains a node $\xi \in \subp$ different from $v$, then the subpaths $\restr[q]{\xi,v}$ and $\restr[\pp]{v,\xi}$ form a cycle, which contradicts the graph being series-parallel.
	We show that we can then reroute flow from $q$ to $q'$.
	Choose 
	$$\delta = \min\{\pflow[q],\, \upFr[{\subp}]\}.$$
	Since the subpath $\subp$ consists only of unsaturated arcs, the chosen $\delta$ is positive.
	{Now consider a flow $f'$ that results form re-routing flow with rate $\delta$ from $q$ to $q'$:
		its path decomposition is $z\colon \paths \to \Q_{+}$ with 
		\[
		z(p) = \begin{cases}
			\pflow[q] - \delta,&\quad p= q,\\
			\pflow[q']+\delta\frac{T-\trtp[q']}{\T-\trtp[q]},&\quad p=q',\\
			\pflow, &\quad\text{otherwise}.
		\end{cases}
		\]
		We verify that $z$ is a feasible flow decomposition:
		the value $\za$ of the corresponding static flow on an arbitrary arc $a$ is
		\[
		\za = \begin{cases}
			\xa - \delta, &\text{ if } a \in q\setminus q',\\
			\xa + \delta\cdot\frac{\T-\trtp[q]}{\T-\trtp[q']}, &\text{ if } a\in q'\setminus q,\\
			\xa -\delta(1 - \frac{\T-\trtp[q]}{\T-\trtp[q']}), &\text{ if } a\in q'\cap q,\\
			\xa, &\text{ if } a \notin q\cup q',
		\end{cases}
		\]
		Hence, for all arcs $a\notin q'\setminus q$ we have $\za \leq \xa$, and for $a \in q'\setminus q$ we obtain 
		\[ \za \leq \xa + \delta \leq  \xa + \upFr[\subp] \leq \ua
		\]
		as $a\in \subp$.
		
		By~\Cref{cl:reroute}, we obtain that  flow $f'$ has either lower \peakcosts than flow $f$, if $q'$ is strictly shorter than $q$, or a higher aggregated amount of flow on the path \pp, if $\trtp[q'] = \trtp[q]$. 
		This is again a contradiction to the choice of flow $f$.
	}
		Hence, our assumption was false and $\restr[q]{v,w} = \subp$ is true.
	This implies that no flow-carrying path takes a detour between nodes $v$ and $w$; that is, any flow-carrying path $p \in \paths$ containing both $v$ and $w$ also contains the subpath $\subp$.
	
	Next, we show that there does exist some saturated arc~$a \in \subp$ with $\xa \geq \ua[\pp]$, which would lead Case~3 to contradiction.
	The maximality of the subpath~$\subp$ implies that either~$v=s$ or there exists a saturated arc~$\aa = (v', v)$ on path $\pp$, and, equivalently, either $w=t$ or there is a saturated arc $\atil = (w,w')$ on \pp.
	Since we assume that path \pp has at least one saturated arc, at least one of the arcs $\aa$ or $\atil$ exists, see~\Cref{fig:unsat-subpath}. 
	
	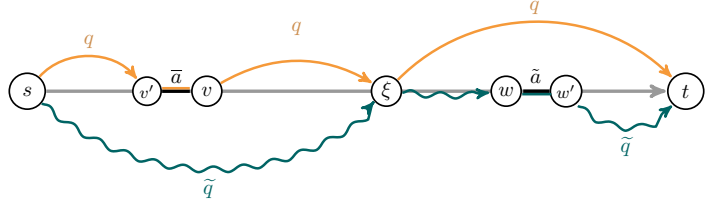
\begin{figure}[tb]
		\centering
		\resizebox{0.6\textwidth}{!}{
			\input{tikz/unsat-subpath}
		}
		\caption{A shortest $s$-$t$ path \pp (in grey) with an unsaturated subpath between the arcs $\overline{a}$ and~$\atil$; paths~$q$ (in orange) and~$\widetilde{q}$ (in petrol wavy) contain the saturated arc $\aa$ (resp.~$\atil$) but not the node $w$ (resp.~$v$).}
		\label{fig:unsat-subpath}
	\end{figure}
	
	Suppose first that every path $p\in \pospaths$ containing arc $\atil$ also contains node $v$; this holds, in particular, also if~$v=s$. 
	Then every such path contains the entire subpath $\subp$,
	so for any arc $a\in \subp$ we have 
	\[\xa = \sum_{p\in \paths\colon a\in p} \pflow \geq \sum_{p\in \paths\colon \atil\in p} \pflow = \xa[\atil] \geq Y^*,\]
	which implies $\xa \geq Y^*$ and contradicts arc $a$ being unsaturated.
	The symmetric case, i.e.~if all paths containing arc $\aa$ also contain node $w$, is identical.
	Therefore, both arcs $\aa$ and $\atil$ exist on path~\pp, and there exists a path~$\widetilde{q}\in \pospaths$ containing arc~$\atil$ but not node~$v$.
	
	Now consider the unsaturated arc $a_1 = (v, v_1) \in \subp$ incident to arc~$\aa$ in path~\pp. 
		We show that any flow-carrying path traversing arc $\aa$ also traverses $a_1$. 
		So let $q\in\pospaths$ be a path containing $\aa$.
		If $w\in q$, then path~$q$ coincides with \pp on the segment $\subp$; in particular, we have $a_1 \in \subp \subset q$.
		Otherwise, if $w\notin q$, then the following holds for paths~$q$ and $\widetilde{q}$ containing arc $\atil$ but not node $v$.
		\begin{claim}\label{cl:commonNode}
			There exists a node $\xi \in \subp $ such that $\xi\in {q} \cap \widetilde{q}$.
		\end{claim}
		The claim follows immediately from \Cref{lem:commonNode}, which is proven at the end of this subsection.

		So let $\xi\in \subp$ be an intersection node of ${q}$ and $\widetilde{q}$.
		Then the path $\restr[{q}]{s,\xi} \circ \restr[\widetilde{q}]{\xi,t}$ is a flow-carrying $s$-$t$ path containing both nodes $v$ and $w$. 
		Consequently, we have $\restr[{q}]{v,\xi} \circ \restr[\widetilde{q}]{\xi,w} = \subp$, and, in particular, $a_1 \in q$.
		
		Hence, any path ${q} \in \pospaths$ containing arc $\aa$ also contains arc~$a_1$, and for the static flow values we obtain $\xa[a_1] \geq \xa[\aa] \geq Y^*$,
		so arc $a_1$ is saturated.	
		This contradiction completes Case~3.
		
		\vspace{2ex}
		Since each of the considered cases leads to a contradiction, we obtain that $U = \emptyset$ and that the shortest path~$\pp$ is saturated along all its length.
	\end{proof}
	
	\Cref{thm:shortestPathFull} shows that there exists a minimum \peakcosts \TR flow that uses a shortest $s$-$t$~path at full capacity. 
	Next, we strengthen this result by showing in  \Cref{thm:pathDecompFull} that there exists an optimal \TR flow with a path decomposition that assigns a chosen shortest path $\pp$ the maximum possible flow rate~$\ua[\pp]$.
	
	Beforehand, we prove the following auxiliary result.
	Observe that if two flow-carrying paths intersect, we can construct two new paths by swapping the parts of the paths between the intersection and node $t$. 
	Clearly, we can also swap several segments simultaneously, if paths intersect at several points.
	We formalise the described swapping the segments and show afterwards that redistributing the flow in the flow decomposition from the old to the new recombined paths yields a feasible \TR flow. 
	Moreover, if the obtained paths have a shorter and a longer transit time than the original paths, then the new flow covers the same flow demand at lower or equal \peakcosts. 
	
	We begin with formalising 
	
	\begin{definition}[Recombined paths]\label{def:segmentSwap}
			Let an acyclic network over time $\net=\netwDef$ be given.
		Let $p_1$ and $p_2$ be two $s$-$t$~paths intersecting at nodes~$I = \{v_1, \ldots, v_k\}\subseteq V$.
		Let the nodes in $I$ be enumerated in the order in which they occur in path $p_1$.
		Then paths 
		\begin{align*} q_1 &\coloneqq \restr[p_1]{s, v_1}\circ \restr[p_2]{v_1, v_2}\circ \ldots \circ \restr[p_{(1+ k\text{\,mod\,} 2)}]{v_k, t} \qquad
		\text{and}\\
		q_2 &\coloneqq \restr[p_2]{s, v_1}\circ \restr[p_1]{v_1, v_2}\circ \ldots \circ \restr[p_{(2-k\text{\,mod\,} 2)}]{v_k, t}
		\end{align*}
		are well-defined simple $s$-$t$~paths, called \emph{recombined paths}.
	\end{definition}
	\begin{lemma}\label{lem:recombPathsfeas}
		Recombined paths as in \Cref{def:segmentSwap} are well-defined and simple.
	\end{lemma}
	\begin{proof}
		First note that both paths $p_1$ and $p_2$ traverse the intersection nodes $v_1,\ldots, v_k$ in the same order - otherwise the graph contains a cycle.
		Hence, every subpath $\restr[p_i]{v_j, v_{j+1}}$ for $i\in \{1,2\}$ and $0\leq j \leq k$ is simple (here we use~$v_0 = s$ and $v_{k+1}=t$ for simpler notation).
		
		Finally, we argue that the recombined paths $q_1$ and $q_2$ are simple by contradiction. 
		Suppose that, without loss of generality, path~$q_1$ contains a node $\xi$ in segments~$\restr[p_1]{v_i, v_{i+1}}$ and~$\restr[p_2]{v_j, v_{j+1}}$ with $0\leq i,j \leq k$ and  $i\neq j$.
		If~$i<j$, then the subpaths $\restr[p_1]{\xi, v_{i+1}}$, $\restr[p_1]{v_{i+1}, v_j}$ and $\restr[p_2]{v_j, \xi}$ build a cycle.
		If~$j<i$, then the subpaths $\restr[p_2]{\xi, v_{j+1}}$, $\restr[p_2]{v_{j+1}, v_i}$ and $\restr[p_1]{v_i, \xi}$ build a cycle.
		Since both cases contradict the fact that the graph is acyclic, we conclude that the path~$q_1$ and, by symmetry, also path~$q_2$, is simple.
	\end{proof}

	\begin{lemma}[Segment swap]\label{lem:segmentSwap}
		Let an acyclic network $\net=\netwDef$ and a \TR flow $f$ with time horizon~$\T$, demand~\dem and path decomposition $\pathdecomp$ be given.
		Let $p_1$ and $p_2$ be two flow-carrying $s$-$t$~paths, and let $q$ and~$\qhat$ be a pair of paths recombined from $p_1$ and $p_2$.
		Further, let the transit times of the paths satisfy  $\trtp[q] \leq \trtp[p_i] \leq \trtp[\qhat]$ for $i\in \{1,2\}$.
		Finally, choose a strictly positive $\delta< \min\{\pflow[p_1],\pflow[p_2]\}$.
		Then a \TR flow $f'$ with time horizon $\T$ corresponding to a path decomposition 
		\[z\colon\paths\to\Q_{+},\qquad p\mapsto \begin{cases}
			\pflow[p] - \delta, &p\in\{p_1,p_2\},\\
			\pflow[p] + \delta, &p\in\{q, \qhat\},\\
			\pflow, &\text{otherwise},
		\end{cases}
		\]
		is a feasible \TR flow of value $\val[f']=\dem$ and with \peakcosts $\cmax[f'] \leq \cmax$.
		
		We call the described transformation of the flow $f$ into $f'$ a \emph{segment swap}.
	\end{lemma}
	\begin{proof}
		First, we argue that the constructed flow is feasible, i.e.~respects arc capacities. 
		The total static flow over an arc $a\in A$ is
		\[ Z(a) =  \sum_{\substack{p\in \paths,\\a\in p}} z(p) = \sum_{\substack{p\in \paths,\\a\in p}} \pflow + \delta\cdot\abs{\{p\in \{q,\qhat\} \mid a\in p\}} - \delta\cdot\abs{\{p\in \{p_1, p_2\} \mid a\in p\}} 
		\overset{\text{(1)}}{=} \sum_{\substack{p\in \paths,\\a\in p}} \pflow \leq \ua,
		\]
		where equality (1) holds since any arc is contained in paths $p_1$ and $p_2$ with the same multiplicity an in paths~$q$ and~$\qhat$, and thus $\abs{\{p\in \{q,\qhat\} \mid a\in p\}} =\abs{\{p\in \{p_1, p_2\} \mid a\in p\}}$.
		
		We thus observe that the corresponding static flow does not change after a segment swap.
		For the flow values, we have
		\begin{align*}
			\val[f'] - \val[f] &= \sum_{{p\in \paths}} \left(z(p)-\pflow\right)\left(\T-\trtp\right) \\
			&=\sum_{p\in \{p_1,p_2\}}-\delta\left(\T-\trtp[p]\right) 
			+ \sum_{p\in \{q,\qhat\}}\delta(\T-\trtp[p])\\
			&= \delta\big(2\T - \trtp[q] - \trtp[\qhat] - 2\T + \trtp[p_1] + \trtp[p_2]\big)\\
			&= 0;
		\end{align*}
		hence, the new flow $f'$ also has value \dem.
		
		Finally, we compare the \peakcosts of the flows.
		The change is the \peakcosts is 
		\[ \Delta c^{\text{max}} \coloneqq \cmax[f'] - \cmax = \delta\big(\w[q] + \w[\qq] - \w[p_1] - \w[p_2] \big).
		\] 
		We assume without loss of generality that $\trtp[p_1]\leq \trtp[p_2]$ and differentiate cases depending on the paths' lengths.
		\begin{description}
			\item[1.] $\trtp \leq \frac{\T}{2}$ for all $p\in \{p_1,p_2,q,\qq\}$.
			Then $\w[p] = \trtp$ for all four paths and $\Delta c^{\text{max}} = 0$.
			\item[2.] $\trtp \leq \frac{\T}{2}$ for $p\in\{q,p_1,p_2\}$, and $\trtp[\qq]>\frac{\T}{2}$.
			Then 
			\[
			\frac{\Delta c^{\text{max}}}{\delta} = \big(\trtp[q] +\T -\trtp[\qq] - \trtp[p_1] - \trtp[p_2]\big) \overset{\trtp[p_1] + \trtp[p_2] = \trtp[q] + \trtp[\qq]}{=} \big( \T - 2\trtp[\qq]\big) <0.
			\]
			\item[3.] $\trtp \leq \frac{\T}{2}$ for $p\in\{q,p_1\}$, and $\trtp>\frac{T}{2}$ for $p\in \{p_2, \qq\}$. Then
			\[
			\frac{\Delta c^{\text{max}}}{\delta} =\trtp[q] +\T -\trtp[\qq] - \trtp[p_1] -\T+ \trtp[p_2] 
			= \big( \trtp[q] - \trtp[p_1]\big) + \big(\trtp[p_2] -\trtp[\qq] \big) \leq 0.
			\]
			\item[4.] $\trtp[q]\leq \frac{\T}{2}$, and $\trtp>\frac{T}{2}$ for $p\in \{p_1, p_2, \qq\}$. Then
			\[
			\frac{\Delta c^{\text{max}}}{\delta} = \trtp[q] + \T - \trtp[\qq] - \T + \trtp[p_1] - \T+\trtp[p_2]
			\overset{\trtp[p_1] + \trtp[p_2] = \trtp[q] + \trtp[\qq]}{=}
			-\T + 2\trtp[q] \leq 0.
			\]
			\item[5.]  $\trtp>\frac{T}{2}$ for all $p\in \{q, p_1, p_2, \qq\}$. 
			Then $\w[p] = \T-\trtp$ for all paths and $\Delta c^{\text{max}} = 0$.
		\end{description}
		
		Hence, for any combination of paths lengths, the change in the \peakcosts is non-positive, i.e.~\mbox{$\cmax[f'] \leq \cmax$} as desired.
		
		Note that we choose $\delta$ to be strictly smaller than each flow rate in order for the old paths $p_1$ and $p_2$ to remain flow-carrying.
	\end{proof}
	
	Equipped with the segment swap operation, we now prove the key statement leading to the correctness of the MSSP algorithm.	
	\begin{theorem}\label{thm:pathDecompFull}
		Let a network $\net=\netwDef$, a time horizon~$\T$, demand~$\dem$, and a shortest $s$-$t$~path~$\pp$ in~$G$ be given. 
		There exists a minimum-peak-cost \TR flow $(f, \pflowvar)$ of value \dem such that its path decomposition $\pathdecomp$ satisfies $\pflow[\pp] = \ytarg \coloneqq \min\{\ua[\pp], \frac{\dem}{\T-\trtp[\pp]}\}$.
	\end{theorem}
	\begin{proof}
		By \Cref{thm:shortestPathFull}, there exists an optimal \TR flow $f$ with path decomposition $\pathdecomp$ that uses path $\pp$ at full capacity, i.e.~such the corresponding static flow is 
		$			\xa \geq \ytarg$
		on every arc $a\in \pp$.
		
		So suppose that, for any such optimal flow, the flow rate over path \pp is $\pflow[\pp] < \ytarg$.
		Then for each arc~$a\in\pp$ we have $\pflow[\pp] <\xa$, so there exists a further flow-carrying path~$p_a\neq \pp$ with~$a\in p_a$.
		The collection of all such paths is a \emph{cover} of path \pp.
		
		Formally, we define a \emph{cover of a path \pp under flow $f$ with path decomposition $\pflowvar$} as 
		a tuple $(Q, V^Q)$, where~$Q$ is a tuple $(p_1,\ldots,p_k)$ of $s$-$t$~paths with $p_i\in \pospaths$ and $V^Q = (v_1,\ldots, v_{k-1})$ is a sequence of \emph{intersection} nodes on path~\pp, for~$k\in \N$, such that
		\begin{itemize}
			\item $v_{i-1}$ lies before $v_i$ on \pp,
			\item path $p_i$ covers the segment $\restr[\pp]{v_{i-1},v_{i}}$ of $\pp$, i.e.~$\restr[\pp]{v_{i-1},v_{i}} \subseteq p_i$,
		\end{itemize}
		for each $i\in \oneto{k}$ with the convention $v_0 = s$ and $v_k=t$.
		Note that the paths in a cover are not necessarily pairwise-disjoint.
		
		Let $\mathcal{F}$ be the set of all \TR flows with the following properties:
		\begin{enumerate}
			\item $\val = \dem$; \label{prop:dem}
			\item the \peakcosts $\cmax = c^*$ is minimum; \label{prop:cmax}
			\item path \pp is used at full capacity, i.e.~$\xa \geq \ytarg$ for each $a\in \pp$;\label{prop:fullCap}
			\item $\pflow[\pp] < \ytarg$.\label{prop:flowrate}
		\end{enumerate}
		As already stated, every flow in $\mathcal{F}$ admits a cover. 
		Let $f \in \mathcal{F}$ be a flow admitting a cover $(Q, V^Q)$ with a minimum number $\abs{Q} = k$ of segments among flows in $\mathcal{F}$; if there are several such flows, let $f$ with path decomposition $y$ be additionally the one maximising the flow rate $\pflow[\pp]$ over the path \pp.
		To prove the theorem by contradiction, we show that there exists another flow $f'$ with path decomposition $z$ satisfying all properties of the set $\mathcal{F}$ and either having lower \peakcosts or admitting a cover with less segments than~$f$.
		
		Observe that for a cover $(Q, V^Q)$ of path \pp, for a path $p_i \in Q$ covering the segment of \pp between nodes~$v_{i-1} \in V^Q$ and~$v_i\in V^Q$ and for any path $p\in \paths$ passing through nodes~$v_{i-1}$ and~$v_i$ we have \[
		\trtp[{\restr[p]{v_{i-1}, v_i}}] \geq	\trtp[{\restr[\pp]{v_{i-1}, v_i}}],
		\] 
		since otherwise the path $\trtp[{\restr[\pp]{s,v_{i-1}}}]\circ\trtp[{\restr[p]{v_{i-1}, v_{i}}}]\circ\trtp[{\restr[\pp]{v_{i},t}}]$ would be shorter than \pp.
		We will be using this observation to swap subpaths of some covering paths in order to obtain a new path decomposition as in~\Cref{lem:segmentSwap}. 
		In the following, we consider different configurations of the paths in the cover and indicate which path segments induce an improving segment swap in each case. 
		To simplify the notation, we write $$\restr[p]{v,w} \leq \restr[q]{v,w}$$ to denote that $\trtp[{\restr[p]{v,w}}] \leq \trtp[{\restr[q]{v,w}}]$.
		
		Consider the first two paths $p_1$ and $p_2$ of the cover, covering subpaths $\restr[\pp]{s,v_1}$ and $\restr[\pp]{v_1,t}$ respectively.
		We know that $\restr[p_1]{s,v_1} \leq \restr[p_2]{s, v_1}$.  
		First, suppose that the opposite relation holds for the complementary subpaths, i.e.~$\restr[p_1]{v_1,t} \geq \restr[p_2]{v_1,t}$.
		Then consider paths 
		\[
		q\coloneqq \restr[p_1]{s,v_1}\circ \restr[p_2]{v_1,t}\qquad \text{and} \qquad \qq\coloneqq \restr[p_2]{s,v_1}\circ \restr[p_1]{v_1,t}.
		\]
		For the transit times we have $\trtp[q] \leq \trtp[p_i]$ and $\trtp[\qq] \geq \trtp[p_i]$ for $i\in \{1,2\}$.
		Hence, by applying a segment swap as in~\Cref{lem:segmentSwap}, we obtain a new \TR flow $(f',z)$ with $\val[f'] = \dem$ and $\cmax[f'] \leq \cmax$.
		If~$\cmax[f']<\cmax$, then $f$ was not an optimal flow, which is a contradiction. 
		Otherwise, if $\cmax[f']=\cmax$, flows $(f,y)$ and $(f',z)$ have also the same corresponding static flow and the same value, so flow $f'$ satisfies properties \ref{prop:dem}-\ref{prop:fullCap} of set~$\mathcal{F}$.
		
		If $q\neq \pp$, then the new flow rate over \pp is $z(\pp) \leq \pflow[\pp] < \ytarg$, so $(f',z)\in \mathcal{F}$.
		Observe that path $q$ now covers both subpaths $\restr[\pp]{s, v_1}$ and $\restr[\pp]{v_1, v_2}$.
		Furthermore, segment swap, by construction, does not remove flow-carrying paths, i.e.~$\pospaths \subset \pospaths[z]$.
		Therefore, the cover $(Q', V^{Q'})$ obtained from $(Q, V^Q)$ by replacing the first two paths by $q$, i.e.~with  
		$Q'\coloneqq (q, p_3,\ldots,p_k)$,
		and $ V^{Q'} \coloneqq (v_2, \ldots, v_{k-1})$
		is a feasible cover of path~\pp under flow~$(f',z)$ and has size~$k-1 < k$,
		which contradicts the choice of flow $f$.
		
		Suppose now that $q=\pp$, which is possible if $k=2$.
		If the path decomposition $z$ of $f'$ satisfies $z(\pp) = \ytarg$, then the theorem statement holds;
		otherwise, we have $z(\pp) < \ytarg$ and thus $f'\in \mathcal{F}$. 
		Then the cover $(Q, V^Q)$ is still a cover of path \pp also under flow $f'$, so the minimum cover under $f'$ not greater than that under $f$. 
		In addition, the flow rate over path \pp is $z(\pp) > \pflow[\pp]$, which again contradicts the choice of flow $f$. 
		
		So far we have considered the case  $\restr[p_1]{v_1,t} \geq \restr[p_2]{v_1,t}$. 
		Now, suppose that  $\restr[p_1]{v_1,t} < \restr[p_2]{v_1,t}$.
		Then the size of the cover is $\abs{Q}\geq 3$, as otherwise we would have $\restr[p_2]{v_1,t} = \restr[\pp]{v_1,t}$ and  $\trtp[p_1] < \trtp[\pp]$.
		
		Let $j\in \oneto{k-1}$ be the smallest index with $\restr[p_j]{v_j,t} \geq \restr[p_{j+1}]{v_j,t}$. Note that the node $v_j$ is contained in both paths).
		Such an index exists, since the inequality is satisfied for $j=k-1$.
		By our assumption we have $j\geq 2$.
		We consider the following cases.
		
		\vspace*{1ex}\noindent\underline{Case 1}: $\restr[p_j]{s,v_j} < \restr[p_{j+1}]{s,v_j}$.\\
		Then consider the paths 
		\[
		q\coloneqq \restr[p_j]{s,v_j} \circ \restr[p_{j+1}]{v_j,t} \qquad \text{and} \qquad \qq\coloneqq \restr[p_{j+1}]{s,v_j} \circ \restr[p_{j}]{v_j,t}.
		\]
		Their transit times satisfy the relation $\trtp[q] \leq \trtp[p_{i}] \leq \trtp[\qq]$ for $i\in \{j,\, j+1\}$. 
		Hence, by applying a segment swap according to~\Cref{lem:segmentSwap} we obtain a new flow $f'$.
		Again, either $f'$ has strictly lower \peakcosts than flow $f$, which contradicts the optimality of $f$, or $f' \in \mathcal{F}$.
		In the latter case, path $q$ covers two subpaths $\restr[\pp]{v_{j-1},v_{j}} = \restr[p_j]{v_{j-1}, v_j}$ and $\restr[\pp]{v_j,v_{j+1}} = \restr[p_{j+1}]{v_{j}, v_{j+1}}$ of path \pp.
		Therefore, the cover $(Q', V^{Q'})$ with 
		\[ 	Q'\coloneqq (p_1,\ldots, p_{j-1}, q, p_{j+2}, \ldots,p_{k}) \quad\text{and}\quad V^{Q'} \coloneqq (v_1,\ldots,v_{j-1}, v_{j+1}, \ldots, t)
		\]
		is a feasible cover for \pp under flow $f'$. 
		Cover $(Q', V^{Q'})$ has size $k-1<\abs{Q}$, which contradicts the choice of flow $f$ from the set $\mathcal{F}$.
		
		\vspace*{1ex}\noindent\underline{Case 2}: $\restr[p_{j-1}]{s,v_{j-1}} \geq \restr[p_{j}]{s,v_{j-1}}$.\\
		By definition of $j$, we additionally know that $\restr[p_{j-1}]{v_{j-1},t} < \restr[p_{j}]{v_{j-1}, t}$.
		Then we can apply a segment swap to recombined paths \[
		q\coloneqq  \restr[p_{j}]{s,v_{j-1}} \circ \restr[p_{j-1}]{v_{j-1},t}  \qquad \text{and} 
		\qquad \qq\coloneqq \restr[p_{j-1}]{s,v_{j-1}} \circ \restr[p_{j}]{v_{j-1},t}
		\]
		with  $\trtp[q] \leq \trtp[p_{i}] \leq \trtp[\qq]$ for $i\in \{j-1, j\}$.
		Hence, either the resulting flow $f'$ has \peakcosts $\cmax[f'] < \cmax$, or $f'\in \mathcal{F}$ and admits a cover $(Q', V^{Q'})$ with
		\[
		Q'\coloneqq (p_1,\ldots, p_{j-2}, \qq, p_{j+1}, \ldots,p_{k}) \quad\text{and}\quad V^{Q'} \coloneqq (v_1,\ldots,v_{j-2}, v_{j}, \ldots, t),
		\]
		which is feasible under flow $f'$ since path $\qq$ covers the subpath $\restr[\pp]{v_{j-2}, v_{j-1}} \circ \restr[\pp]{v_{j-1}, v_{j}}$.
		In both cases we obtain a contradiction to the choice of flow $f$.
		
		It remains to consider the remaining complementary case.	
		\vspace*{1ex}\paragraph*{\underline{\textnormal{Case 3:}}}\label{mosterproof:case3} $\restr[p_{j-1}]{s,v_{j-1}} < \restr[p_{j}]{s,v_{j-1}}$ and $\restr[p_j]{s,v_j} \geq \restr[p_{j+1}]{s,v_j}$.\\
		Simultaneously, by the definition of $j$ we have 
		\begin{equation}\label{eq:def_j}
			\restr[p_{j-1}]{v_{j-1},t} < \restr[p_{j}]{v_{j-1}, t}\quad\text{ and }\quad\restr[p_j]{v_j,t} \geq \restr[p_{j+1}]{v_j,t}.
		\end{equation}
		
		In this case, none of the previous path recombinations would lead to an improvement of \peakcosts. 
		However, we show that the paths $p_{j-1}$, $p_j$ and $p_{j+1}$ have a further common intersection node, which allows further swaps.
		
		\begin{claim}\label{cl:commonNodeOnPj}		
			There exists a node $\xi\in \restr[\pp]{v_{j-1}, v_j}$ such that $\xi \in p_{j-1} \cap p_{j+1}$ (see \Cref{fig:relevantPaths}).
		\end{claim}
		Since paths $p_{j-1}$ and $p_{j+1}$ intersect~\pp at nodes~$v_{j-1}$ and~$v_{j}$, respectively, the claim holds due to \Cref{lem:commonNode}, which is proved at the end of this section.
		Observe that $\xi \in \restr[p_j]{v_{j-1}, v_j}$; that is, all four paths intersect at node~$\xi$.
		
		\begin{figure}
			\centering
			\resizebox{\textwidth}{!}{
				\input{tikz/relevant-paths-ssp}
			}
			\caption{Paths $p_{j-1}$, $p_j$ and $p_{j+1}$ all intersect at node $\xi\in \pp$.}
			\label{fig:relevantPaths}
		\end{figure}
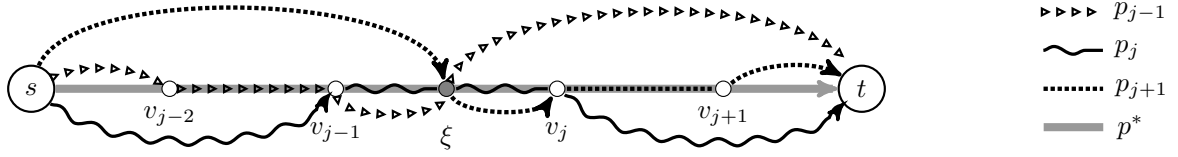
		
		We first consider the two edge cases of node $\xi$ being one of the cover intersection nodes $v_{j-1}$ or $v_j$.
		
		\noindent\underline{Case 3.a}: $\xi = v_{j}$.
		Then paths $p_{j-1}$ and $p_j$ intersect at nodes $v_{j-1}$ and $v_j$.
		Since graph $G$ is acyclic, both paths the nodes in the same order as path \pp, i.e.~subpaths $\restr[p_{j-1}]{v_{j-1}, v_j}$ and $\restr[p_{j}]{v_{j-1}, v_j}$ are well-defined.
		Hence, we define the following segment swap:
		\[
		q\coloneqq  \restr[p_{j-1}]{s,v_{j-1}} \circ \restr[p_{j}]{v_{j-1},v_j}\circ\restr[p_{j-1}]{v_j, t}  \qquad \text{and} 
		\qquad \qq\coloneqq \restr[p_{j}]{s,v_{j-1}} \circ \restr[p_{j-1}]{v_{j-1},v_j}\circ\restr[p_{j}]{v_j, t}.
		\]
		
		By assumptions of Case 3, the segment-wise comparison of transit times yields that $\trtp[q] \leq \trtp[p_i] \leq \trtp[\qq]$	for~$i\in \{j-1, j\}$.
		Hence, the resulting flow has lower \peakcosts or, since path~$q$ covers the subpath~$\restr[\pp]{v_{j-2}, v_j}$, the resulting flow admits a cover $(Q', V^{Q'})$ with $\abs{Q}-1$ paths, e.g.~the cover
		\[
		Q'\coloneqq (p_1,\ldots, p_{j-2}, q, p_{j+1}, \ldots,p_{k}) \quad\text{and}\quad V^{Q'} \coloneqq (v_1,\ldots,v_{j-2}, v_{j}, \ldots, t).
		\]
		
		\noindent\underline{Case 3.b}: $\xi = v_{j-1}$.
		In this case, paths $p_j$ and $p_{j+1}$ intersect at two nodes $v_{j-1}$ and $v_j$, which allows us to make a segment swap with paths
		\[
		q\coloneqq  \restr[p_{j+1}]{s,v_{j-1}} \circ \restr[p_{j}]{v_{j-1},v_j}\circ\restr[p_{j+1}]{v_j, t}  \qquad \text{and} 
		\qquad \qq\coloneqq \restr[p_{j}]{s,v_{j-1}} \circ \restr[p_{j+1}]{v_{j-1},v_j}\circ\restr[p_{j}]{v_j, t}
		\]
		with  $\trtp[q] \leq \trtp[p_i] \leq \trtp[\qq]$	for $i\in \{j, j+1\}$.
		Since path $q$ covers the subpath $\restr[\pp]{v_{j-1}, v_{j+1}}$, the resulting flow admits a cover $(Q', V^{Q'})$ for path $\pp$ with \[
		Q'\coloneqq (p_1,\ldots, p_{j-1}, q, p_{j+2}, \ldots,p_{k}) \quad\text{and}\quad V^{Q'} \coloneqq (v_1,\ldots,v_{j-1}, v_{j+1}, \ldots, t)
		\]
		and thus with $\abs{Q'}<\abs{Q}$, which is again a contradiction to the choice of flow $f$.
		
		Finally, we consider the remaining case of $\xi$ being an inner point of the subpath $\restr[\pp]{v_{j-1}, v_j}$.
		
		\noindent\underline{Case 3.c}: $\xi \notin \{v_{j-1}, v_j\}$.
		Here, we perform a combination of two segment swaps (cf.~\Cref{fig:recombPaths}).
		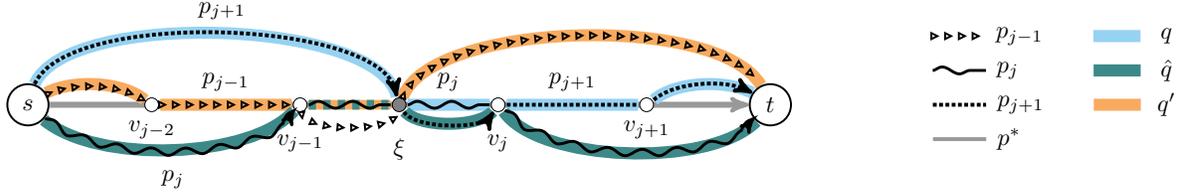
\begin{figure}
			\centering
			\resizebox{\textwidth}{!}{
				\input{tikz/recombined-paths}
			}
			\caption{Recombined paths: recombining~$p_{j}$ and~$p_{j+1}$ yields paths~$\hat{q}$ (bright blue) and~$q'$ (dark petrol); in the second step, a segment swap between $\hat{q}$ and $p_{j-1}$ yields the path~$q'$ (orange). Paths $q$ and $q'$ yield a smaller cover of the segment $\restr[\pp]{v_{j-2}, v_{j+1}}$.}
			\label{fig:recombPaths}
		\end{figure}
		
		For the first swap, consider the recombination to paths
		\[
		q\coloneqq  \restr[p_{j+1}]{s,\xi} \circ \restr[p_{j}]{\xi,v_j}\circ\restr[p_{j+1}]{v_j, t}  \qquad \text{and} 
		\qquad \qq\coloneqq \restr[p_{j}]{s,\xi} \circ \restr[p_{j+1}]{\xi,v_j}\circ\restr[p_{j}]{v_j, t}.
		\]
		First, we ensure that the transit times satisfy the condition of \Cref{lem:segmentSwap}. We have
		\begin{align*}
			\restr[q]{s, v_j} = \restr[p_{j+1}]{s, \xi}\circ \restr[p_j]{\xi, v_j} 
			\overset{\restr[p_j]{\xi, v_j} = \restr[\pp]{\xi, v_j}}{\leq} &\restr[p_{j+1}]{s, v_j} \\
			\overset{\text{Case 3}}{\leq} \quad\quad&\restr[p_{j}]{s, v_j}\\
			\overset{\restr[p_j]{\xi, v_j} = \restr[\pp]{\xi, v_j}}{\leq }&\restr[p_{j}]{s, \xi}\circ \restr[p_{j+1}]{\xi, v_j}
			= \restr[\qq]{s, v_j}
		\end{align*}
		and
		\begin{align*}
			\restr[q]{v_j,t} = \restr[p_{j+1}]{v_j, t} \overset{\eqref{eq:def_j}}{\leq}\restr[p_{j}]{v_j, t} = \restr[\qq]{v_j, t}.
		\end{align*}
		Hence, the transit times satisfy the relation $\trtp[q] \leq \trtp[p_i] \leq \trtp[\qq]$	for $i\in \{j, j+1\}$.
		Moreover, path $q$ covers the subpath $\restr[\pp]{\xi, v_{j+1}}$.
		The flow resulting from this segment swap thus does not lead to a smaller cover. 
		We nevertheless switch to flow $f'$ and cover  $(Q', V^{Q'})$ with 
		\[
		Q'\coloneqq (p_1,\ldots, p_{j}, q, p_{j+2}, \ldots,p_{k}) \quad\text{and}\quad V^{Q'} \coloneqq (v_1,\ldots,v_{j-1}, \xi, v_{j+1}, 	\ldots, t).
		\]
		Next, we consider a second segment swap between paths $p_{j-1}$ and $\qq$: let
		\[
		q'\coloneqq \restr[p_{j-1}]{s, v_{j-1}}\circ\restr[\qq]{v_{j-1, \xi}} \circ\restr[p_{j-1}]{\xi, t}   \qquad \text{and} \qquad
		q''\coloneqq \restr[\qq]{s, v_{j-1}}\circ\restr[p_{j-1}]{v_{j-1, \xi}} \circ\restr[\qq]{\xi, t}.
		\]
		We again verify the condition of \Cref{lem:segmentSwap}:
		\begin{align*}
			\restr[q']{s, v_{j-1}} = \restr[p_{j-1}]{s, v_{j-1}} \overset{\text{Case 3}}{<} \restr[p_j]{s, v_{j-1}} =\restr[\qq]{s, v_{j-1}}= \restr[q'']{s, v_{j-1}}
		\end{align*}
		and
		\begin{align*}
			\restr[q']{v_{j-1}, t} = \restr[\qq]{v_{j-1}, \xi}\circ \restr[p_{j-1}]{\xi, t} \quad 
			= \quad &\restr[p_j]{v_{j-1}, \xi}\circ \restr[p_{j-1}]{\xi, t}\\
			\overset{\restr[p_j]{v_{j-1},\xi} = \restr[\pp]{v_{j-1},\xi}}{\leq} &\restr[p_{j-1}]{v_{j-1}, t} \\
			\overset{\eqref{eq:def_j}}{<}\qquad \quad &\restr[p_{j}]{v_{j-1}, t}\\
			\overset{\restr[p_j]{v_{j-1},\xi} = \restr[\pp]{v_{j-1},\xi}}{\leq} &\restr[p_{j-1}]{v_{j-1}, \xi}\circ\restr[p_j]{\xi,v_j}\circ \restr[p_{j}]{v_{j}, t} \\
			\overset{\restr[p_j]{\xi, v_j} = \restr[\pp]{\xi, v_j}}{\leq}\quad &\restr[p_{j-1}]{v_{j-1}, \xi}\circ\restr[p_{j+1}]{\xi,v_j}\circ \restr[p_{j}]{v_{j}, t}
			= \restr[q'']{v_{j-1},t}.
		\end{align*}
		
		Hence, we again have  $\trtp[q'] \leq \trtp[p_{j-1}] \leq \trtp[q'']$;
		Finally, since 
		$
		\restr[p_{j-1}]{v_{j-1},t} < \restr[p_{j}]{v_{j-1},t}
		$ by definition of $j$ and~$\restr[p_{j-1}]{v_{j-1},\xi} \geq \restr[p_{j}]{v_{j-1},\xi}$ due to the cover property, we have 
		$\restr[p_{j-1}]{\xi, t} < \restr[p_{j}]{\xi, t}$, and thus 
		\[
		\restr[q']{\xi, t} = \restr[p_{j-1}]{\xi, t} \overset{}{\leq} \restr[p_{j}]{\xi, t} \leq \restr[p_{j+1}]{\xi, v_j}\circ \restr[p_{j}]{v_j, t} = \restr[\qq]{\xi,t}.
		\]
		Therefore, we also have  $\trtp[q'] \leq \trtp[\qq] \leq \trtp[q'']$, so the segment swap satisfies the condition of \Cref{lem:segmentSwap}.
		Let flow~$f''$ be the flow resulting form the second segment swap.
		Path~$q'$ covers the subpath $\restr[\pp]{v_{j-2}, \xi}$. 
		Hence, a cover $(Q'', V^{Q''})$ with 
		\[
		Q''\coloneqq (p_1,\ldots, p_{j-2}, q', q, p_{j+2}, \ldots,p_{k}) \quad\text{and}\quad V^{Q'} \coloneqq (v_1,\ldots,v_{j-2}, \xi, v_{j+1}, 	\ldots, t)
		\]
		is a feasible cover under flow $f''$ of size $\abs{Q''} = k-1 < \abs{Q}$.
		Therefore, flow~$f''$ has lower \peakcosts or a smaller cover for path \pp, so we again obtain a contradiction to the choice of flow~$f$.

		In each considered cases we were able to construct a new flow that contradicts the minimality of flow~$f$. 
		Therefore, the assumption about existence of a flow in set $\mathcal{F}$ is false, so $\mathcal{F} = \emptyset$.
		Consequently, every minimum \peakcosts \TR flow of requested value~\dem that uses path~\pp at full capacity also has a path decomposition~$\pflowvar$ with $\pflow[\pp] = \ua[\pp]$, and by \Cref{thm:shortestPathFull} there always exists such a flow. 
	\end{proof}
	
	Finally, we use the result of~\Cref{thm:pathDecompFull} to prove the correctness of \algoname.
	\begin{theorem}\label{thm:greedyCorrect}
		The \algoname algorithm finds a \TR flow covering required demand~\dem at minimum \peakcosts on unit-cost two-terminal series-parallel graphs.
	\end{theorem}
	\begin{proof}
		We prove the statement by an induction over the demand value~\dem.
		For $\dem=0$ and $\dem=1$, the proof follows from \Cref{thm:shortestPathFull}.
		
		Assume that \algoname computes an optimal flow for demand values not greater than~$\dem-1$.
		Let an instance~$(\net, \T, \dem)$ of \genProb be given, and let~$\pp$ be the shortest path selected by the algorithm in the first iteration.
		If~$\dem \leq \ua[\pp] \cdot (\T-\trtp[\pp])$, then the algorithm stops after the first iteration, and its output is optimal by~\Cref{thm:shortestPathFull}. 
		
		So suppose that $\dem > \ua[\pp] \cdot (\T-\trtp[\pp])$.
		Now consider the \emph{reduced instance} $(\net', \T, \dem')$ of \genProb with capacity along path \pp and demand both reduced by a chain flow along \pp. 
		The instance is formally defined as follows:
		\[	\net' \coloneqq (G, \trt[], \ua[]', \ca[]) \quad \text{with} \quad \ua' = \begin{cases}
			\ua-\ua[\pp],&\text{if }a \in \pp,\\
			\ua, &\text{otherwise};
		\end{cases}   \quad\qquad \dem' \coloneqq \dem-\ua[\pp]\big(\T-\trtp[\pp]\big).
		\]
		Observe that it is exactly the instance on which the algorithm works starting at the second iteration.
		As~\mbox{$\dem'<\dem$}, \algoname finds an optimal flow~$f'$ for the reduced instance by the induction hypothesis.
		Observe that the flow~\mbox{$f \coloneqq f' + f^{\T}_{\pp}$} constructed by adding a chain flow with flow rate~$\ua[\pp]$ along path~$\pp$ is then exactly the output of \algoname on the original instance.
		It remains to show that flow $f$ is optimal.
		
		Let $f^*$ be an optimal flow for the original instance $(\net, \T, \dem)$ that uses path $\pp$ at full capacity, i.e.~for its path decomposition $\pflowvar^*$ holds $\pflowvar^*(\pp) = \ua[\pp]$.
		Now consider a reduced flow~$f^*_{\net'}$ corresponding to the path decomposition 
		\[
		\pflowvar^*_{\net'} \colon \paths \to \R_{+},\quad p\mapsto \begin{cases}
			\pflowvar^*(p), &\text{if }p\neq \pp,\\
			0, &\text{if }p=\pp,
		\end{cases}
		\]
		and time horizon \T.
		The reduced flow is a feasible \TR flow for the instance $(\net', \T, \dem')$ and has \peakcosts
		\[\cmax[f^*_{\net'}] = \sum_{p\in \paths,\,p\neq \pp}  \pflowvar^*_{\net'}(p)\cdot \w = \cmax[f^*] - \ua[\pp]\w[\pp].
		\]
		
		At the same time, since flow~$f'$ is an optimal flow for the reduced instance, we have
		\[ \cmax[f] =  \cmax[f'] + \ua[\pp]\w[\pp] \leq \cmax[f^*_{\net'}]  + \ua[\pp]\w[\pp] \leq \cmax[f^*],
		\]
		so flow~$f$ has minimum \peakcosts.
	\end{proof}
	
	We finish this section with a proof of existence of the intersection nodes as in \Cref{cl:commonNode}.
	\newcommand{\paux}{p}
	\newcommand{\ppaux}{q}
	\newcommand{\vaux}{\overline{v}}
	\newcommand{\waux}{\overline{w}}
	\begin{lemma}\label{lem:commonNode}
		Let $G$ be a two-terminal series-parallel graph, and let $\pp$, $\paux$ and $\ppaux$ be $s$-$t$ paths.
		Let $\paux$ and $\ppaux$ intersect path $\pp$ at nodes $v$ and $w$, respectively, 
		such that $v$ precedes $w$ on path \pp, i.e.~$w\in \restr[\pp]{v,t}$.
		Then there exists a node $\xi \in \subp$ such that $\xi\in \paux \cap \ppaux$.
	\end{lemma}
	\begin{proof}
		Suppose that $v\notin \ppaux$ and $w\notin \paux$, as otherwise the statement is immediately proven.
		
		Let $\vaux\in \paux$ be the last node on the path $\paux$ with $\vaux \in \subp$; that is, the subpath $\restr[\pp]{\vaux,w}$ contains no nodes from path $\paux$ except for $\vaux$. 
		Note that $\vaux \neq w$.
		Further, let $\waux \in \ppaux$ be the first node on the subpath $\restr[\pp]{\vaux, w}$ belonging to path $\ppaux$, so that the subpath $\restr[\pp]{\vaux, \waux}$ contains no nodes from path $\ppaux$ except for $\waux$; see~\Cref{fig:non-sp-paths}.
		
		\begin{figure}[bt]
			\centering
			\resizebox{0.6\textwidth}{!}{
				\input{tikz/non-sp}
			}
		\caption{Path \pp is intersected by paths $p$ and $q$. Paths~$p$ (in orange) and~$q$ (in petrol wavy) are disjoint with the subpath of \pp between $\overline{v}$ and $\overline{w}$.}
			\label{fig:non-sp-paths}
		\end{figure}
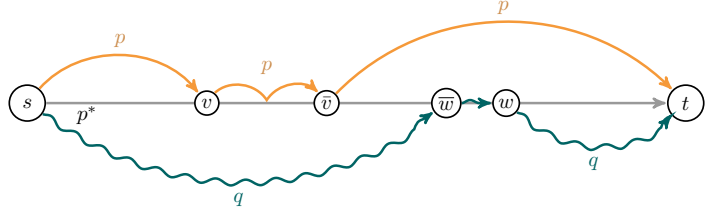
	
		Suppose that the statement is false, i.e.~paths $\paux$ and $\ppaux$ do not intersect on $\subp$.
		In particular, this implies~$\overline{v} \neq \overline{w}$.
		Then consider a subgraph $G'$ of $G$ induced by the following three paths: 
		\[\pp = (\restr[\pp]{s,\vaux} \circ \restr[\pp]{\vaux, \waux} \circ \restr[\pp]{\waux, t}), \quad \restr[\paux]{\vaux, t}\quad\text{ and }\quad\restr[\ppaux]{s, \waux}.\]
		Since 
		$\restr[\pp]{\vaux, \waux}\cap \ppaux = \{\waux\}$ as well as $\restr[\pp]{\vaux, \waux}\cap \paux = \{\vaux\}$,
		the subgraph $G'$ is homeomorphic to the graph $H = (V', A')$ with $V' =\{s', \widetilde{v}, \widetilde{w}, t'\}$ and arcs $A' = \big\{
		(s', \widetilde{v}), (\widetilde{v},\widetilde{w}), (\widetilde{w}, t'), (s', \widetilde{w}), (\widetilde{v}, t')\big\}$, which is a forbidden subgraph for series-parallel graphs \cite{valdesTL-SPgraph}.
		This contradiction to graph $G$ being series-parallel proves the lemma.
	\end{proof}

\subsection{General acyclic networks}
For non-series-parallel but acyclic networks, fractional \genProb on unit-cost networks is still solvable in polynomial time via column generation.

\begin{theorem}
	Fractional \genProb on unit-cost acyclic networks is solvable in polynomial time.
\end{theorem}
\begin{proof}
Recall that the \peakcosts of a \TR flow can be expressed as 
$
\cmax = \sum_{{p\in \paths}} \w \cdot \pflow
$,
where $\w = \min\{\trtp,\ \T-\trtp\}$ for any path $p\in \paths$.
Hence, we express \genProb as a linear programme (LP) with variables~$y_p$ denoting the flow rates along paths:
\begin{align*}\tag{P}\label{lp2:prim}
	\text{min. } & \sum_{p\in \paths} \w y_p && \\
	\text{s.t. } &  \sum_{p\in\paths,\, a\in p} y_p \leq \ua && \forall a\in A && \mid \pi_a \\
	& \sum_{p\in \paths} \big(\T-\trtp\big) y_p \geq \dem && && \mid z \\
	& y_p \geq 0 && \forall p\in \paths. &&
\end{align*}
We introduce dual variables $\pi_a$ for capacity constraints and a variable $z$ for the demand constraint, and obtain the following dual LP:
\begin{align*}\tag{D}\label{lp2:dual}
	\text{max. } & \dem z -\sum_{a\in A} \ua\pi_a \\
	\text{s.t. } &  (\T-\trtp)\cdot z -\sum_{a\in p} \pi_a \leq \w && \forall p\in \paths\\
	& \pi_a \geq 0 && \forall a \in A\\
	& z\geq 0.
\end{align*}

The pricing problem is now 
\begin{align*}
	\min_{p\in \paths}\ \w -  \big(\T-\trtp\big) z + \suma \pi_a &=
	\min \begin{cases}
		\min_{p\in \paths}\,  \trtp -  \big(\T-\trtp\big) z + \suma \pi_a,\\
		\min_{p\in \paths}\, \T - \trtp -  \big(\T-\trtp\big) z + \suma \pi_a\\
	\end{cases}\\
	&= \min \begin{cases}
		\min_{p\in \paths}\, \suma \big(\pi_a + (z+1) \trt\big) - \T z,\\
		\min_{p\in \paths}\, \suma \big(\pi_a + (z-1) \trt\big) - \T (z-1).\\
	\end{cases}
\end{align*}
Both inner minimisation subproblems can be interpreted as a Shortest Path problem on the instance's network. 
Although the auxiliary arc costs may be negative in the second subproblem, both subproblems are solvable in polynomial time on acyclic networks. 
\end{proof}

%% file: tikz/unsat-subpath.tex
\pgfdeclarelayer{bg}
\pgfsetlayers{bg, main}

\colorlet{aux-green}{Orange!50!combi-orange}
\colorlet{aux-blue}{combi-darkcyan}

\begin{tikzpicture}[>=stealth',shorten >=1pt, shorten <=0.5pt, auto, node distance = 1.5cm]
	\tikzstyle{every state}=[thick, inner sep=0.5mm, minimum size=6mm, outer sep=0mm]
	\tikzstyle{every edge} = [draw, ->]
	\tikzset{internode/.style={circle,fill, black, inner sep=0.mm, label[below]}}
	\tikzset{greenstyle/.style={draw=aux-green, very thick}}
	\tikzset{bluestyle/.style={draw=aux-blue, very thick, decorate, decoration={snake, segment length=5mm, amplitude=0.5mm}}}
	
	\node[state] (s) at (0, 0) {$s$};
	\node[state] (t) at (11, 0) {$t$};
	\begin{pgfonlayer}{bg}
		\path[draw=black!40, ultra thick, ->] (s) -- (t);
	\end{pgfonlayer}
	
	\node[state,  minimum size=4mm, fill=white] (vpr) at (2, 0) {\footnotesize${v'}$};
	
	\node[state,  minimum size=5mm, fill=white] (v) at (3, 0) {${v}$};
	\path[draw, ultra thick] (vpr) --node[above]{$\overline{a}$} (v.west);
	\node[state,  minimum size=5mm, fill=white] (vbar) at (6, 0) {$\xi$};
	
	\node[state,  minimum size=4mm, fill=white] (wpr) at (9, 0) {\footnotesize$w'$};
	\node[state,  minimum size=5mm, fill=white] (w) at (8, 0) {$w$};
	\path[draw, ultra thick] (wpr) --node[above]{$\tilde{a}$} (w.east);
	
	\draw[->, greenstyle] (s) to[bend left=50] node[above, text=aux-green!80!black]{$q$} (vpr);
	\draw[greenstyle] ($(vpr.east)+(0,0.05)$)--($(v.west)+(0,0.05)$);
	\draw[greenstyle,->] (v) to[bend left=30] node[above=3mm, text=aux-green!80!black]{${q}$} (vbar);
	\draw[greenstyle,->] (vbar) to[bend left=45] node[above, text=aux-green!80!black]{${q}$} (t);
	
	\draw[bluestyle, ->] (s) to[bend right=40]node[below=1mm, text=aux-blue]{$\widetilde q$}  (vbar);
	\draw[bluestyle, ->] (vbar) to[bend right = 5] (w);
	\draw[draw=aux-blue, very thick, ] ($(w.east)+(0,-0.05)$) -- ($(wpr.west) + (0.05,-0.05)$);
	\draw[bluestyle, ->] ($(wpr.south east) + (0,-0.05)$) to[bend right=40]node[below=1mm, text=aux-blue]{$\widetilde q$} (t);

\end{tikzpicture}

%% file: tikz/relevant-paths-ssp.tex
\pgfdeclarelayer{bg}
\pgfsetlayers{bg, main}

\colorlet{aux-green}{LimeGreen}
\colorlet{aux-blue}{Cerulean}
\colorlet{aux-red}{Orange}

\begin{tikzpicture}[>=stealth',shorten >=0.5pt, shorten <=0.5pt, auto, node distance = 1.5cm, xscale=1.8]
	\tikzstyle{every state}=[thick, inner sep=0.5mm, minimum size=6mm, outer sep=0mm]
	\tikzstyle{every edge} = [draw, ->]
	\tikzstyle{hili} = [draw, line width=5pt]
	\tikzset{internode/.style={draw, fill=white,circle, minimum size=2mm, inner sep=0.mm}}
	\tikzset{greenstyle/.style={draw=aux-green!70, line width=5pt}}
	\tikzset{bluestyle/.style={draw=aux-blue!30, line width=5pt}}
	\tikzset{redstyle/.style = {draw=aux-red!50, line width=5pt}}
	\tikzset{dashstyle/.style = {draw, thick, decorate, decoration={triangles, segment length=2mm}}}
	\tikzset{wavestyle/.style = {draw, very thick, decorate, decoration={snake, segment length=5mm, amplitude=0.5mm}}}
	\tikzset{dotstyle/.style = {draw, ultra thick, densely dotted}}
	
	\node[state] (s) at (0, 0) {$s$};
	\node[state] (t) at (6, 0) {$t$};
		\path[draw=black!40, line width=3pt] (s) -- ($(t) + (-0.2, 0)$);
		\path[draw=black!40, ultra thick, ->] ($(t) + (-0.2, 0)$) -- (t);
	
	\node[internode, label=below:{$v_{j-2}$}] (v-2) at (1,0) {};
	\node[internode, label={[label distance=2mm]below:{$v_{j-1}$}}] (v-1) at (2.2,0) {};
	\node[internode, fill=black!50, label={[label distance=2.5mm]below:{$\xi$}}] (xi) at (3,0) {};
	\node[internode, label={[label distance=2mm]below:{$v_{j}$}}] (v) at (3.8,0) {};
	\node[internode, label=below:{$v_{j+1}$}] (vv) at (5,0) {};
	
	\draw[dotstyle, ->] (s)  to[bend left=80] (xi);
	\draw[dotstyle, ->] (xi)  to[bend right=70] (v);
	\draw[dotstyle] (v) -- (vv);
	\draw[dotstyle, ->](vv) to[bend left=50]  (t);
	
	\draw[dashstyle, ->] (s)  to[bend left=45] (v-2);
	\draw[dashstyle] (v-2) -- (v-1);
	\draw[dashstyle, ->] (v-1) to[bend right=85] (xi);
	\draw[dashstyle, ->] (xi)  to[bend left=75] (t);
	
	\draw[wavestyle, ->] (s) to[bend right=55] (v-1);
	\draw[wavestyle] (v-1) -- (xi);
	\draw[wavestyle] (xi) -- (v);
	\draw[wavestyle, ->] (v)  to[bend right=55] (t);
	
	\node[] (p1) at (7.9,1.){};
	\draw[dashstyle]  ($(p1) + (-0.6,0)$) -- ++(0.45, 0) node[right]{$p_{j-1}$};
	\node[] (p2) at ($(p1) + (0,-0.5)$){};
	\draw[wavestyle] ($(p2) + (-0.6,0)$) -- ++(0.45, 0) node[right]{$p_j$};
	\node[] (p3) at ($(p2) + (0,-0.5)$) {};
	\draw[dotstyle] ($(p3) + (-0.6,0)$) -- ++(0.45, 0) node[right] {$p_{j+1}$};
	\node[] (p4) at ($(p3) + (0,-0.5)$) {};
	\draw[draw=black!40, line width=3pt] ($(p4) + (-0.6,0)$) -- ++(0.45, 0) node[right] {$\pp$};
	
\end{tikzpicture}

%% file: tikz/recombined-paths.tex
\pgfdeclarelayer{bg}
\pgfsetlayers{bg, main}

\colorlet{aux-red}{Orange!50!combi-orange}
\colorlet{aux-blue}{Cerulean}
\colorlet{aux-green}{combi-darkcyan!80}

\begin{tikzpicture}[>=stealth',shorten >=0.5pt, shorten <=0.5pt, auto, node distance = 1.5cm, xscale=1.8]
	\tikzstyle{every state}=[thick, inner sep=0.5mm, minimum size=6mm, outer sep=0mm]
	\tikzstyle{every edge} = [draw, ->]
	\tikzstyle{hili} = [draw, line width=5pt]
	\tikzset{internode/.style={draw, fill=white,circle, minimum size=2mm, inner sep=0.mm}}
	\tikzset{greenstyle/.style={draw=aux-green!95, line width=5pt}}
	\tikzset{bluestyle/.style={draw=aux-blue!40, line width=5pt}}
	\tikzset{redstyle/.style = {draw=aux-red!80, line width=5pt}}
\tikzset{dashstyle/.style = {draw, thick, decorate, decoration={triangles,  segment length=2mm}}}
\tikzset{wavestyle/.style = {draw, very thick, decorate, decoration={snake, segment length=5mm, amplitude=0.5mm}}}
\tikzset{dotstyle/.style = {draw, ultra thick, densely dotted}}

\node[state] (s) at (0, 0) {$s$};
\node[state] (t) at (6, 0) {$t$};
	\begin{pgfonlayer}{bg}
\path[draw=black!40, line width=2.5pt] (s) -- ($(t) + (-0.2, 0)$);
\path[draw=black!40, ultra thick, ->] ($(t) + (-0.2, 0)$) -- (t);
	\end{pgfonlayer}

\node[internode, label=below:{$v_{j-2}$}] (v-2) at (1,0) {};
\node[internode, label={[label distance=2mm]below:{$v_{j-1}$}}] (v-1) at (2.2,0) {};
\node[internode, fill=black!50, label={[label distance=2.5mm]below:{$\xi$}}] (xi) at (3,0) {};
\node[internode, label={[label distance=2mm]below:{$v_{j}$}}] (v) at (3.8,0) {};
\node[internode, label=below:{$v_{j+1}$}] (vv) at (5,0) {};

\begin{pgfonlayer}{bg}
\draw[bluestyle] (s) to[bend left=80] node[above]{$p_{j+1}$} (xi);
\draw[bluestyle] (xi) -- node[above]{$p_{j}$} (v);
\draw[bluestyle] (v) -- node[above]{$p_{j+1}$} (vv);
\draw[bluestyle] (vv) to[bend left=50] (t);

\draw[redstyle] (s)  to[bend left=50] (v-2);
\draw[redstyle] (v-2)  --node[above]{$p_{j-1}$}  (v-1);
\draw[redstyle,line width=4pt] (v-1)  -- (xi);
\draw[redstyle] (xi)  to[bend left=75] (t);
\draw[greenstyle] (s)  to[bend right=60] node[below=1mm]{$p_{j}$} (v-1);
\draw[greenstyle, line width=4pt, dashed] (v-1) -- (xi);
\draw[greenstyle] (xi)  to[bend right=65] (v);
\draw[greenstyle] (v)  to[bend right=65] (t);	
\end{pgfonlayer}

\draw[dotstyle, ->] (s)  to[bend left=80] (xi);
\draw[dotstyle, ->] (xi)  to[bend right=70] (v);
\draw[dotstyle] (v) -- (vv);
\draw[dotstyle, ->](vv) to[bend left=50]  (t);

\draw[dashstyle, ->] (s)  to[bend left=45] (v-2);
\draw[dashstyle] (v-2) -- (v-1);
\draw[dashstyle, ->] (v-1) to[bend right=85] (xi);
\draw[dashstyle, ->] (xi)  to[bend left=75] (t);

\draw[wavestyle, ->] (s) to[bend right=55] (v-1);
\draw[wavestyle] (v-1) -- (xi);
\draw[wavestyle] (xi) -- (v);
\draw[wavestyle, ->] (v)  to[bend right=55] (t);

\node[] (p1) at (7.9,1.){};
\draw[dashstyle]  ($(p1) + (-0.6,0)$) -- ++(0.45, 0) node[right]{$p_{j-1}$};
\node[] (p2) at ($(p1) + (0,-0.5)$){};
\draw[wavestyle] ($(p2) + (-0.6,0)$) -- ++(0.45, 0) node[right]{$p_j$};
\node[] (p3) at ($(p2) + (0,-0.5)$) {};
\draw[dotstyle] ($(p3) + (-0.6,0)$) -- ++(0.45, 0) node[right] {$p_{j+1}$};
\node[] (p4) at ($(p3) + (0,-0.5)$) {};
\draw[draw=black!40, ultra thick] ($(p4) + (-0.6,0)$) -- ++(0.45, 0) node[right] {$\pp$};

\node (q1) at (9.2,1.) {$q$};
\draw[bluestyle]  ($(q1) + (-0.6,0)$) -- ++(0.4, 0);
\node (q2) at ($(q1) + (0,-0.5)$) {$\hat{q}$};
\draw[greenstyle] ($(q2) + (-0.6,0)$) -- ++(0.4, 0);
\node (q3) at ($(q2) + (0,-0.5)$) {$q'$};
\draw[redstyle] ($(q3) + (-0.6,0)$) -- ++(0.4, 0);
\end{tikzpicture}

%% file: tikz/non-sp.tex
	\pgfdeclarelayer{bg}
	\pgfsetlayers{bg, main}
	
	\colorlet{aux-green}{Orange!50!combi-orange}
	\colorlet{aux-blue}{combi-darkcyan}
	
	\begin{tikzpicture}[>=stealth',shorten >=1pt, shorten <=0.5pt, auto, node distance = 1.5cm]
		\tikzstyle{every state}=[thick, inner sep=0.5mm, minimum size=6mm, outer sep=0mm]
		\tikzstyle{every edge} = [draw, ->]
		\tikzset{internode/.style={circle,fill, black, inner sep=0.mm, label[below]}}
		\tikzset{greenstyle/.style={draw=aux-green, very thick}}
		\tikzset{bluestyle/.style={draw=aux-blue, very thick, decorate, decoration={snake, segment length=5mm, amplitude=0.5mm}}}
		
		\node[state] (s) at (0, 0) {$s$};
		\node[state] (t) at (11, 0) {$t$};
		\begin{pgfonlayer}{bg}
			\path[draw=black!40, very thick, ->] (s) -- (t);
		\end{pgfonlayer}
		
		\coordinate (as) at (2, 0);
		\node at ($(s) + (1, -0.2)$) {$\pp$};
		\node[state,  minimum size=4mm, fill=white] (v) at (3, 0) {${v}$};
		\node[state,  minimum size=4mm, fill=white] (vbar) at (5, 0) {$\bar{v}$};
		
		\coordinate (atilt) at (9,0);
		\node[state,  minimum size=4mm, fill=white] (w) at (8, 0) {$w$};
		\node[state,  minimum size=4mm, fill=white] (vtil) at (7, 0) {$\overline{w}$};
		
		\draw[->, greenstyle] (s) to[bend left=50] node[above, text=aux-green!80!black]{$p$} (v); 

		\draw[greenstyle,->] (v) to[bend left=50] (4,0.05) node[above=3mm, text=aux-green!80!black]{${p}$} to[bend left=50] (vbar);
		\draw[greenstyle,->] (vbar) to[bend left=45] node[above, text=aux-green!80!black]{${p}$} (t);
		
		\draw[bluestyle, ->] (s) to[bend right=35]node[below=1mm, text=aux-blue]{$q$}  (vtil);
		\draw[bluestyle, ->] (vtil) -- (w);
		\draw[bluestyle, ->] (w) to[bend right=40]node[below=1mm, text=aux-blue]{$q$} (t);

	\end{tikzpicture}

%% file: sections/generalLongT.tex
\section{Long time horizons}\label{sec:longT}
\renewcommand{\xa}{x(a)}
In the instance constructed in the proof of \Cref{lem:int-strongNPhard}, transit times of all $s$-$t$ paths in the network are almost as long as the time horizon.
 In this section, we consider the opposite case, where the time horizon is significantly longer than all simple $s$-$t$ paths in the network.  
 We say that the time horizon is \emph{long} if \T is such that $\T \geq 2\trtp$ for any~$s$-$t$ path~$p$. 
 {Note that, given an instance of \genProb, it is in general \np-hard to determine whether this condition holds. 
 	In practical scenarios, however, the lengths of the paths are often known by design or can be sufficiently bounded from above.}
 
 Instances with long time horizons have a useful structural property: any \TR flow also attains its peak cost in the middle of the time horizon, and we again obtain a closed formula for the peak cost.
 
    \begin{lemma}\label{lem:costForHalfT}
 	Let $\net = \netwDef$ be a network and $\T\in \N$ a time horizon such that all paths in~$\net$ have 
 	transit time not greater than $ \frac{\T}{2}$.
 	Then a maximum \TR flow $f$ associated with a static flow $x$ 
 	has \peakcosts
 	\[
 	\cmax = \sum_{a\in A} \ca \cdot \trt \cdot \xa,
 	\]
 	which is attained at time $\extp=\left\lfloor\frac{T}{2}\right\rfloor$.
 \end{lemma}
 \begin{proof}
 	First, observe that the term $\suma \ca \cdot \trt \cdot \xa$ is an upper bound on \instantcosts of a temporally repeated flow corresponding to $x$.
 	It is thus enough to show that this cost is indeed attained at time $\extp$.
 	
 	Since transit times are integral, any path $p$ in the network, and thus in the path decomposition, has transit time $\trtp \leq \extp$.
 	Consider an arbitrary arc~$a=(v,w)\in G$ and a path $p \in \paths$ in the path decomposition of flow $f$ such that $\pflow>0$ and~$a\in p$.  
 	Then we have $p \in \actpaths{\auxtp}$ for each $\auxtp \in [\extp-\trt,\extp)$ (see Definition~\ref{def:trf}), as
 	\[
 	\trtp[p\vert_{s,v}] \leq \trtp - \trt \leq \extp - \trt \leq \auxtp
 	\]
 	and 
 	\[
 	\trtp[p\vert_{w,t}] \leq \trtp \leq \extp  = \left\lfloor\frac{\T}{2}\right\rfloor \leq \T - \extp \overset{\auxtp < \extp} < \T - \auxtp.
 	\]
 	Hence, every path of the flow $f$ containing arc $a$ uses this arc in time interval $[\extp - \trt, \extp)$, and the flow rate on arc $a$ for any~$\auxtp \in [\extp - \trt, \extp)$ is
 	\[
 	f_a(\auxtp) = \sum_{p\in \actpaths{\auxtp}} \pflow = \sum_{\substack{p\in \paths,\\ a\in p}} \pflow = \xa.       \]
 	The cost of flow $f$ at time $\extp$ is thus
 	\[
 	c(f, \extp) = \suma \ca\cdot \int_{\extp-\trt}^{\extp} f_a(\auxtp)\, d\auxtp = \suma \ca\cdot \xa \cdot (\extp - \extp + \trt) = \suma \ca\cdot\trt\cdot\xa = \cmax.
 	\]
 \end{proof}
This structural property is again not sufficient to make the integral problem polynomial-time solvable.
\begin{theorem}\label{thm:longT-weakNP}
	Integral \genProb is at least weakly \np-hard if the time horizon is long. 
\end{theorem}
\begin{proof}
	\newcommand{\cspc}{\kappa}
	We reduce the Constrained Shortest Path problem (CSP) to \genProb.
	An instance $\inst$ of CSP is given by a directed graph $G=(V,A)$ with two designated nodes $s$ and $t$, non-negative arc costs $\cspc_a$ and lengths $\ell_a$ for each arc $a\in A$, and a length bound $L\in\N_0$. 
	The problem asks for an $s$-$t$ path in $G$ with length not greater than~$L$ and with minimum cost.
	
	Given an instance of CSP, we construct an instance of \genProb as follows: we obtain a network~\mbox{$G' = (V', A')$} by adding to $G$ a node $s'$ that is connected to node $s$ by an arc $(s',s)$.
	We set all arc capacities to~$1$.
	The new arc $(s',s)$ has cost~$0$ and transit time~$0$.
	For all other arcs~$a\in A'$, the transit time in the length of their counterpart arc in~$G$, and the cost is $\frac{\cspc_a}{\ell_a}$.
	To obtain integral arc parameters without loss of generality, all costs and transit times can be multiplied by the least common multiplier $\text{lcm}(\ell_a)_{a\in A}$.
	We choose the time horizon $\T\coloneqq 2\sum_{a\in A'}\trt$ to ensure that the network satisfies the condition of the theorem, and set the demand to $\dem\coloneqq \T-L$.
	
	Since every $s'$-$t$ path in~$G'$ contains arc~$(s',s)$, and since its capacity is one, any integral temporally-repeated flow in the constructed network~$\net'$ uses only one path. 
	By \Cref{lem:costForHalfT}, the \peakcosts of a \TR flow using path $p^*$ is 
	\[
		\suma \ca\cdot\trt\cdot\xa  = \sum_{a\in p^*} \ca\cdot\trt = \sum_{a\in p^*, a\neq (s',s)} \frac{\cspc_a}{\ell_a}\cdot\ell_a = \cspc(p^*),
	\]
	which is the cost of path $p^* \setminus\{(s',s)\}$ in the CSP instance.
	Hence, there is a one-to-one correspondence between feasible \TR flows with peak cost not greater than some value~$z$ and $s$-$t$ paths in graph~$G$ with cost at most~$z$ and length not greater than~$L$, so the two instances are equivalent.
\end{proof}

Surprisingly, the problem becomes easy again if the we ask for a flow of maximum possible value.

\input{sections/longTmax}


Relaxing the integrality requirement for flow rates makes the problem easy in this second special case as well.
\begin{theorem}\label{thm:cg}
	Fractional \genProb with long time horizon is solvable in polynomial time.
\end{theorem} 
\begin{proof}
	Let \pathdecomp be the path decomposition corresponding to the sought optimal flow. 
	We express the \genProb problem as a linear programme over non-negative real variables~$y_p$ representing the flow rates~$\pflow$ for~$p\in \paths$.
	The objective is \[
	\cmax = \suma \ca\cdot\trt\cdot\xa = \suma \ca\cdot\trt\cdot \sum_{p\in\paths,\, a\in p} y_p = \sum_{p\in \paths} (\sum_{a\in p} \ca\trt) y_p.
	\]
	The flow rates have to respect the arc capacity. i.e. $\xa\leq \ua$ for all arcs $a\in A$, and the flow value $\val$ must be at least \dem (due to the minimisation objective, any optimal solution will have value of exactly \dem).
	
	Putting things together, we obtain the following linear programme:
	\begin{align*}\tag{P}\label{lp:prim}
		\text{min. } & \sum_{p\in \paths} (\sum_{a\in p} \ca\trt) y_p\\
		\text{s.t. } &  \sum_{p\in\paths,\, a\in p} y_p \leq \ua && \forall a\in A\\
		& \sum_{p\in \paths} (\T-\trtp)\cdot y_p \geq \dem\\
		& y_p \geq 0 && \forall p\in \paths.
	\end{align*}
	This primal LP~\eqref{lp:prim} has an exponential number of variables in worst case, which motivates the column generation approach.
	Using dual variables $\pi_a$ for capacity constraints and a variable $z$ for the demand constraint, we obtain the following dual LP:
	\begin{align*}\tag{D}\label{lp:dual}
		\text{max. } & \dem z -\sum_{a\in A} \ua\pi_a \\
		\text{s.t. } &  (\T-\trtp)z -\sum_{a\in p} \pi_a \leq \sum_{a\in p}\ca\trt && \forall p\in \paths\\
		& \pi_a \geq 0 && \forall a \in A\\
		& z\geq 0.
	\end{align*}
	The pricing problem for given values of the dual variables is thus 
	\[ \min_{p\in \paths} \sum_{a\in p}\ca\trt - (\T-\trtp)z +\sum_{a\in p}\pi_a 
	= - \T z + \min_{p\in \paths} \sum_{a\in p} \ca\trt + z\trt + \pi_a,
	\]
	which is solved by computing a shortest path on the graph~$G$ with non-negative arc costs $ \ca\trt + z\trt + \pi_a$ for each arc~$a$.
	Since the pricing problem, i.e.~the separation problem for the dual \eqref{lp:dual}, is solved in polynomial time, so is also the primal LP \eqref{lp:prim}.
\end{proof}

%% file: sections/longTmax.tex
    \begin{theorem}\label{thm:shortPathsPoly}
		A \emph{maximum} \TR flow with minimum \peakcosts on instances with a long time horizon can be found in strongly polynomial time.
    
	\end{theorem}
    We prove \Cref{thm:shortPathsPoly} in two steps.
    First, using the expression for the \peakcosts from \Cref{lem:costForHalfT}, we show in \Cref{lem:strangeTransf} that minimising the \peakcosts is equivalent to maximising a function dependent only on the corresponding static flow and on network parameters. 
    Second, in \Cref{lem:minCostCirc},
    we transform the maximisation of the latter function into a minimum-cost circulation problem on an auxiliary network.
    Since the minimum-cost circulation problem is polynomial-time solvable, we obtain a polynomial algorithm for \genProb with long time horizon.

    We start by establishing a relation between the \peakcosts of a flow over time and the underlying static flow.
    
    \begin{lemma}\label{lem:strangeTransf}
        For a network $\netwDef$, define a number $M\coloneqq \sum_{a \in A} \ca \cdot \trt \cdot \ua +1$.
        Let $\T\in \N$ be the time horizon.
        Let $x'$  be a feasible static flow that admits a \shortdec path decomposition and
        that maximises the term  
        \begin{equation}\label{eq:term-to-max}
            \Phi(x) \coloneqq M\cdot \T \cdot \val[x] - \suma \big( M + \ca\big)\cdot \trt \cdot \xa.
        \end{equation}
        Then any associated \TR flow $f'$ is a maximum temporally repeated flow, and the static flow $x'$ minimises the sum $\sum_{a\in A} \ca \cdot \trt \cdot \xa$ among all static flows that induce maximum \TR flows.
    \end{lemma}
    \begin{proof}
    Let $f$ be a \TR flow associated with a static flow $x$. 
    We transform term \eqref{eq:term-to-max} as follows:
    \begin{align*}
        \Phi(x)= & M\cdot \T \cdot \val[x] - \suma \big( M + \ca\big)\cdot \trt \cdot \xa \\
   = &M\cdot \T \cdot \val[x] - M\cdot \suma \trt \cdot \xa - \suma \ca \cdot\trt\cdot \xa \\
   \overset{\Cref{lem:TRflowValue}}{=} & M \cdot \val[f] - \suma \ca \cdot\trt\cdot \xa.
    \end{align*}

    Suppose a static flow $x'$ maximises expression \eqref{eq:term-to-max}, but the associated \TR flow $f'$ is not maximum. 
    Then there exists a \TR flow $f^{\prime\prime}$ that corresponds to a static flow $x^{\prime\prime}$ and whose value is strictly greater than the value of~$f'$, i.e.~$\val[f''] > \val[f']$.
    
    Since all arc parameters are integers, so are all flow values: for static flows this follows from the main result of Ford and Fulkerson on maximum flows over time \cite{FF1958}, and for associated flows over time from \Cref{lem:TRflowValue}.
    Hence, we have $\val[f''] \geq \val[f'] + 1$, and obtain 
    \begin{align*}
        \Phi(x'') &= M\cdot \val[f''] - \suma \ca \cdot\trt\cdot x''(a)\\
        &\geq M\cdot \val[f'] + M - \suma \ca \cdot\trt\cdot x''(a)\\
        &\geq M\cdot \val[f'] - \suma \ca \cdot \trt \cdot x'(a) + M -\suma \ca \cdot\trt\cdot (x''(a) - x'(a))\\
        &\geq M\cdot \val[f'] - \suma \ca \cdot \trt \cdot x'(a) + M - \suma \ca \cdot\trt\cdot\ua\\
        & > M\cdot \val[f'] - \suma \ca \cdot \trt \cdot x'(a) \\
        &= \Phi(x'),
    \end{align*}
    so flow $x'$ does not maximise term~\eqref{eq:term-to-max}, which is a contradiction.
    Hence, flow $x'$ induces a maximum \TR flow,
    and thus maximises the value $M\cdot \val$. 
    Furthermore, since $x'$ maximises $\Phi(x)$, it 
    has a minimal value of the sum $\suma \ca \cdot\trt \cdot {x}(a)$
    among all static flows ${x}$ inducing maximum \TR flows.
    \end{proof}

A static flow maximising expression \eqref{eq:term-to-max} is found via an auxiliary minimum cost circulation problem, similar to the one used for finding maximum temporally repeated flows \cite{FF1958}. 

\begin{lemma}\label{lem:minCostCirc}
    Let a network $\net = \netwDef$ and a time horizon $\T$ be given.
    Then a static flow $x$ in network $\net$ that maximises expression~\eqref{eq:term-to-max} and has a \shortdec path decomposition can be found in polynomial time.
\end{lemma}
\begin{proof}

\newcommand{\circu}{\overline{x}}

    We transform the graph $G$ into an auxiliary graph $\auxG$ by adding an arc $(t,s)$ with capacity $u_{(t,s)} = \infty$.
    For a number~$M\in \N$ defined as in \Cref{lem:strangeTransf},
    we define arc costs in network $\auxG$ as follows:
    \[
    \auxcost \colon A(\auxG) \to \Z, \quad a\mapsto
    \begin{cases}
        -M\cdot \T,\quad&a=(t,s),\\
        M\cdot \trt + \ca \cdot \trt,\quad&\text{otherwise}.
    \end{cases}
    \tag{$**$}\label{eq:aux-cost}
    \]
    
    Let $\circu$ be a circulation in $\auxG$.
    A static $s$-$t$ flow~$x$ in network $\net$ \emph{corresponding} to $\circu$  is a flow that results from~$\circu$ by removing the flow over the arc $(t,s)$.
    Observe that the value of the circulation and of the corresponding flow 
    is equal to the flow value on the arc $(t,s)$ -- the only ingoing arc of the source $s$.
    
    Now let $\circu$ be a minimum cost circulation in $\auxG$ and $x$ the corresponding flow in $G$.
    We show that $x$ maximises expression \eqref{eq:term-to-max} and that any of its flow decompositions is \shortdec. 
    
    We express the costs of circulation $\circu$ as 
    \begin{align*}
        \auxcost(\circu) &= - M\cdot \T \cdot \circu(t,s)  + \suma \auxcost(a) \cdot \circu(a)\\
        &=  - M\cdot \T \cdot \circu(t,s)  + \suma \auxcost(a) \cdot \xa \\
        &= - M\cdot \T \cdot \val[x] +  \suma \auxcost(a) \cdot \xa\\
         &= - \Phi(x).
    \end{align*}
    Hence, a minimum cost circulation yields a static flow $x$ that maximises term~\eqref{eq:term-to-max}.
    
    It remains to show that flow $x$ admits a $\T$-bounded path decomposition and thus produces a feasible \TR flow.
    Let $\flowdecomp$ be an arbitrary flow decomposition of the flow $x$ of size at most $\abs{A}$.
    Such a decomposition exists and can be computed in polynomial time by the well-known Flow Decomposition Theorem~\cite{AMO1993}.
    We show that we can transform the flow decomposition $y$ into a \shortdec path decomposition $y'$ in linear time, 
    i.e.~a decomposition for which $y'(p) = 0$ for any path~$p\in \paths$ with $\trtp > \T$ and any cycle $p\in \cycles$.

    First, let $p\in \paths$ be an $s$-$t$ path in $G$ with $\trtp>\T$ and a positive flow value $\pflow > 0$.
    Then 
    \[\auxcost(p) = \sum_{a\in p} \auxcost(a) \geq M \cdot \sum_{a\in p}\trt = M\cdot \trtp > M\cdot\T.
    \]
    Since $\pflow >0$, there exists a backward path $\overleftarrow{p}$ in the residual network of the circulation $\circu$ such that the cycle~$\overleftarrow{p} \cup (s,t)$ has costs
    \[\auxcost(\overleftarrow{p}\cup(s,t)) < -M\cdot\T + M\cdot\T = 0.
    \]
    Hence, there is a negative-cost cycle in the residual graph, which contradicts the minimality of $\circu$. 
    So we have~\mbox{$\trtp\leq \T$} for any path $p\in \paths$ with $\pflow > 0$.

    If $p\in \cycles$ is a cycle in $G$ with $\pflow >0$ and $\auxcost(p) \neq 0$, then either $p$ or its reverse $\overleftarrow{p}$ is a negative cycle in the residual network, which contradicts the optimality of circulation $\overline{x}$.
    Hence, the cycle $p$ has cost $\auxcost(p)=0$, and the flow along $p$ can be removed without changing the value or the cost of the circulation.
    Hence, we obtain a \shortdec path decomposition $\pflowvar' \colon \paths \to \R$ 
    by setting $\pflowvar'(p) \coloneqq0$ for all cycles $p\in \cycles$ and $\pflowvar'(p) \coloneqq \pflow$ otherwise.

\end{proof}

Observe that if all arc parameters are integers, then the resulting static flow $x$ is integral, and a path decomposition obtained by, for instance, a greedy Edmonds-Karp heuristic is integral as well.
Hence, the static flow found in \Cref{lem:minCostCirc} yields an integral \TR flow.

Overall, instances of \genProb with sufficiently long time horizons are solved by the following steps:
\begin{enumerate}
    \item Compute $M\coloneqq \sum_{a \in A} \ca \cdot \trt \cdot \ua +1$.
    \item Construct an auxiliary network $\auxG = \big(V(G),\ A(G)\cup (t,s)\big)$ with $\ua[(t,s)]=\infty$ 
    and with an arc cost function as in~\eqref{eq:aux-cost}.
    \item Find a minimum-cost circulation $\overline{x}$ in $\auxG$ and the corresponding static flow $x$ in the original network $G$.
    \item Compute an integral path decomposition of the flow $x$.
    \item The path decomposition yields an integral maximum \TR flow with minimum \peakcosts.
\end{enumerate}

Note that again, as in \Cref{sec:unitCosts}, the described procedure is not enforcing the solution to be integral, but yields integral solutions nonetheless. 
This implies that on instances of \genProb with a long time horizon and maximum demand, we always construct an integral optimal solution.

%% file: sections/outlook.tex
\section{Conclusion and outlook}
\label{sec:conclusion}
In this work, we introduced peak cost as a novel objective for flows over time, which is relevant in scenarios with limited resources that execute transportation. 
Since solutions with a simple structure are of particular interest in such settings, 
 we restricted the solution space to temporally repeated flows and formulated the \genProb problem. 
We showed that this restriction comes at a cost of  the objective value: while for some flow over time problems, for example maximum flows or quickest min-cost flow problem, \TR flows yield optimal solutions or constant factor approximation, we showed that the corresponding approximation ratio is unbounded for the minimum-peak-cost objective.

Similarly to the minimum-cost objective, the integral decision version of \genProb is strongly $\np$-hard, even for two-terminal series parallel graphs with unit transit times, capacities, and costs equal to zero or one.
This implies that the integral optimisation version is strongly $\np$-hard, even under the above restrictions. 

However, we indicated two special classes of instances for which we obtain polynomial algorithms constructing optimal solutions. 
For unit cost networks, we showed that an optimal solution on series-parallel graphs can be found by a greedy algorithm proposed by Ruzika et al.~\cite{RuzikaSS11} for earliest arrival flows, if we allow fractional flow rates or if the flow demand is the maximum possible one, i.e.~if we are asking for a maximum flow with minimum peak cost. 
This fact contrasts with the \np-hardness of the problem for already two different arc cost values, even for a maximum demand and on series-parallel graphs.
In addition, fractional solutions in acyclic unit-cost networks can also be found in weakly-polynomial time via column generation. 
For the special case of long time horizons, we again derived polynomial-time algorithms for the fractional version of the problem via column generation and for the integral version with maximum demand via computing a static minimum cost circulation in an auxiliary graph, similarly to the approach of Ford and Fulkerson for maximum flows over time.

There are multiple avenues for future work: for example, the complexity of the fractional version of \genProb in general case is still not known.
In addition, our results consider only the theoretical computational complexity. 
Hence, the tractability of exact solution methods or heuristics for solving \genProb in a real-world setting shall be explored in future research. 
One promising approach are path-based integer programming formulations, especially if the number of paths is bounded or a branch-and-price algorithm is employed.

%% file: flows_main.bbl
\begin{thebibliography}{10}

\bibitem{AhrensMA}
Emma Ahrens.
\newblock {G}eneralized temporally repeated flows for the quickest
  transshipment and related problems.
\newblock Masterarbeit, RWTH Aachen University, Aachen, 2024.
\newblock \href {https://doi.org/10.18154/RWTH-2024-07726}
  {\path{doi:10.18154/RWTH-2024-07726}}.

\bibitem{AMO1993}
Ravindra~K. Ahuja, Thomas~L. Magnanti, and James~B. Orlin.
\newblock {\em Network flows: theory, algorithms, and applications}.
\newblock Prentice-Hall, Inc., USA, 1993.

\bibitem{Anapolska25}
Mariia Anapolska, Emma Ahrens, Christina Büsing, Felix Engelhardt, Timo
  Gersing, Corinna Mathwieser, Sabrina Schmitz, and Sophia Wrede.
\newblock Minimum-peak-cost flows over time.
\newblock {\em Networks}, n/a(n/a).
\newblock \href {https://doi.org/10.1002/net.70001}
  {\path{doi:10.1002/net.70001}}.

\bibitem{aronson89}
Janine Aronson.
\newblock A survey of dynamic network flows.
\newblock {\em Annals of Operations Research}, 20:1--66, 12 1989.
\newblock \href {https://doi.org/10.1007/BF02216922}
  {\path{doi:10.1007/BF02216922}}.

\bibitem{BeinBT85}
Wolfgang~W. Bein, Peter Brucker, and Arie Tamir.
\newblock Minimum cost flow algorithms for series-parallel networks.
\newblock {\em Discret. Appl. Math.}, 10(2):117--124, 1985.
\newblock \href {https://doi.org/10.1016/0166-218X(85)90006-X}
  {\path{doi:10.1016/0166-218X(85)90006-X}}.

\bibitem{BodlaenderTDL08}
Hans~L. Bodlaender, Richard~B. Tan, Thomas~C. van Dijk, and Jan van Leeuwen.
\newblock {Integer Maximum Flow in Wireless Sensor Networks with Energy
  Constraint}.
\newblock In {\em Scandinavian Workshop on Algorithm Theory, {SWAT}}, pages
  102--113, 2008.
\newblock \href {https://doi.org/10.1007/978-3-540-69903-3\_11}
  {\path{doi:10.1007/978-3-540-69903-3\_11}}.

\bibitem{BurkardDK93}
Rainer~E. Burkard, Karin Dlaska, and Bettina Klinz.
\newblock The quickest flow problem.
\newblock {\em {ZOR} Methods Model. Oper. Res.}, 37(1):31--58, 1993.
\newblock \href {https://doi.org/10.1007/BF01415527}
  {\path{doi:10.1007/BF01415527}}.

\bibitem{FeketeHKK08}
S{\'{a}}ndor~P. Fekete, Alexander Hall, Ekkehard K{\"{o}}hler, and Alexander
  Kr{\"{o}}ller.
\newblock {The Maximum Energy-Constrained Dynamic Flow Problem}.
\newblock In {\em Scandinavian Workshop on Algorithm Theory {SWAT}}, pages
  114--126, 2008.
\newblock \href {https://doi.org/10.1007/978-3-540-69903-3 \_12}
  {\path{doi:10.1007/978-3-540-69903-3 \_12}}.

\bibitem{FS2003}
Lisa Fleischer and Martin Skutella.
\newblock Minimum cost flows over time without intermediate storage.
\newblock In {\em Symposium on Discrete Algorithms, {SODA}}, page 66–75,
  2003.
\newblock URL: \url{https://dl.acm.org/doi/10.5555/644108.644118}.

\bibitem{FleischerS07-quickest}
Lisa Fleischer and Martin Skutella.
\newblock Quickest flows over time.
\newblock {\em {SIAM} J. Comput.}, 36(6):1600--1630, 2007.
\newblock \href {https://doi.org/10.1137/S0097539703427215}
  {\path{doi:10.1137/S0097539703427215}}.

\bibitem{FleischerTardos98}
Lisa Fleischer and {\'{E}}va Tardos.
\newblock Efficient continuous-time dynamic network flow algorithms.
\newblock {\em Oper. Res. Lett.}, 23(3-5):71--80, 1998.
\newblock \href {https://doi.org/10.1016/S0167-6377(98)00037-6}
  {\path{doi:10.1016/S0167-6377(98)00037-6}}.

\bibitem{FF1958}
L.~R. Ford and D.~R. Fulkerson.
\newblock {Constructing Maximal Dynamic Flows from Static Flows}.
\newblock {\em Operations Research}, 6(3):419--433, 1958.
\newblock \href {https://doi.org/10.1287/opre.6.3.419}
  {\path{doi:10.1287/opre.6.3.419}}.

\bibitem{FFbook}
L.~R. Ford and D.~R. Fulkerson.
\newblock {\em {Flows in Networks}}.
\newblock Princeton University Press, 1962.

\bibitem{Gale1959}
David Gale.
\newblock Transient flows in networks.
\newblock {\em Michigan Mathematical Journal}, 6:59--63, 1959.
\newblock \href {https://doi.org/10.1307/mmj/1028998140}
  {\path{doi:10.1307/mmj/1028998140}}.

\bibitem{GottschalkKLPSW18}
Corinna Gottschalk, Arie M. C.~A. Koster, Frauke Liers, Britta Peis, Daniel
  Schmand, and Andreas Wierz.
\newblock Robust flows over time: models and complexity results.
\newblock {\em Math. Program.}, 171(1-2):55--85, 2018.
\newblock \href {https://doi.org/10.1007/S10107-017-1170-3}
  {\path{doi:10.1007/S10107-017-1170-3}}.

\bibitem{HT1987}
Horst~W. Hamacher and Suleyman Tufekci.
\newblock On the use of lexicographic min cost flows in evacuation modeling.
\newblock {\em Naval Research Logistics (NRL)}, 34(4):487--503, 1987.
\newblock \href
  {https://doi.org/10.1002/1520-6750(198708)34:4<487::AID-NAV3220340404>3.0.CO;2-9}
  {\path{doi:10.1002/1520-6750(198708)34:4<487::AID-NAV3220340404>3.0.CO;2-9}}.

\bibitem{HT2000}
Bruce Hoppe and Eva Tardos.
\newblock The quickest transshipment problem.
\newblock {\em Mathematics of Operations Research}, 25(1):36--62, 2000.
\newblock URL: \url{http://www.jstor.org/stable/3690422}.

\bibitem{KW04}
Bettina Klinz and Gerhard~J. Woeginger.
\newblock {Minimum-cost dynamic flows: The series-parallel case}.
\newblock {\em Networks}, 43(3):153--162, 2004.
\newblock \href {https://doi.org/10.1002/net.10112}
  {\path{doi:10.1002/net.10112}}.

\bibitem{KoehlerS05}
Ekkehard K{\"{o}}hler and Martin Skutella.
\newblock Flows over time with load-dependent transit times.
\newblock {\em {SIAM} J. Optim.}, 15(4):1185--1202, 2005.
\newblock \href {https://doi.org/10.1137/S1052623403432645}
  {\path{doi:10.1137/S1052623403432645}}.

\bibitem{Minieka73}
Edward Minieka.
\newblock {Maximal, Lexicographic, and Dynamic Network Flows}.
\newblock {\em Operations Research}, 21(2):517--527, 1973.
\newblock \href {https://doi.org/10.1287/OPRE.21.2.517}
  {\path{doi:10.1287/OPRE.21.2.517}}.

\bibitem{ParpaleaCiurea11}
Mircea Parpalea and E.~Ciurea.
\newblock The quickest maximum dynamic flow of minimum cost.
\newblock {\em International Journal of Applied Mathematics and Informatics},
  3:266--274, 01 2011.

\bibitem{Philpott90}
Andrew~B. Philpott.
\newblock {Continuous-Time Flows in Networks}.
\newblock {\em Math. Oper. Res.}, 15(4):640--661, 1990.
\newblock \href {https://doi.org/10.1287/MOOR.15.4.640}
  {\path{doi:10.1287/MOOR.15.4.640}}.

\bibitem{POWELL1995}
Warren~B. Powell, Patrick Jaillet, and Amedeo Odoni.
\newblock Chapter 3 stochastic and dynamic networks and routing.
\newblock In {\em Network Routing}, volume~8 of {\em Handbooks in Operations
  Research and Management Science}, pages 141--295. Elsevier, 1995.
\newblock \href {https://doi.org/10.1016/S0927-0507(05)80107-0}
  {\path{doi:10.1016/S0927-0507(05)80107-0}}.

\bibitem{RuzikaSS11}
Stefan Ruzika, Heike Sperber, and Mechthild Steiner.
\newblock Earliest arrival flows on series-parallel graphs.
\newblock {\em Networks}, 57(2):169--173, 2011.
\newblock \href {https://doi.org/10.1002/NET.20398}
  {\path{doi:10.1002/NET.20398}}.

\bibitem{SchloeterSk17}
Miriam Schl{\"{o}}ter and Martin Skutella.
\newblock Fast and memory-efficient algorithms for evacuation problems.
\newblock In {\em {ACM-SIAM} Symposium on Discrete Algorithms, {SODA} 2017},
  pages 821--840, 2017.
\newblock \href {https://doi.org/10.1137/1.9781611974782.52}
  {\path{doi:10.1137/1.9781611974782.52}}.

\bibitem{Skutella2009}
Martin Skutella.
\newblock {An Introduction to Network Flows over Time}.
\newblock In {\em {Research Trends in Combinatorial Optimization, Bonn Workshop
  on Combinatorial Optimization, November 3-7, 2008, Bonn, Germany}}, pages
  451--482. Springer, 2008.
\newblock \href {https://doi.org/10.1007/978-3-540-76796-1 \_21}
  {\path{doi:10.1007/978-3-540-76796-1 \_21}}.

\bibitem{SKUTELLA2023_quickMinTP}
Martin Skutella.
\newblock A note on the quickest minimum cost transshipment problem.
\newblock {\em Operations Research Letters}, 51(3):255--258, 2023.
\newblock \href {https://doi.org/10.1016/j.orl.2023.03.005}
  {\path{doi:10.1016/j.orl.2023.03.005}}.

\bibitem{Skutella2024}
Martin Skutella.
\newblock {\em An Introduction to Transshipments Over Time}, page 239–270.
\newblock Cambridge University Press, 2024.
\newblock \href {https://doi.org/10.1017/9781009490559.009}
  {\path{doi:10.1017/9781009490559.009}}.

\bibitem{valdesTL-SPgraph}
Jacobo Valdes, Robert~E. Tarjan, and Eugene~L. Lawler.
\newblock The recognition of series parallel digraphs.
\newblock In {\em Proceedings of the Eleventh Annual ACM Symposium on Theory of
  Computing}, STOC, page 1–12, 1979.
\newblock \href {https://doi.org/10.1145/800135.804393}
  {\path{doi:10.1145/800135.804393}}.

\bibitem{Williamson_2019}
David~P. Williamson.
\newblock {\em Network Flow Algorithms}.
\newblock Cambridge University Press, 2019.

\end{thebibliography}
